%% file: flat_pert.tex
\documentclass[11pt]{article}

\usepackage{fullpage}

\usepackage{url}
\usepackage{enumitem}

\input{environmentsRD}

\input{gencommands}

\usepackage[pagebackref=true]{hyperref}

\title{Delaunay stability via perturbations
}

\author{
Jean-Daniel Boissonnat
\footnote{
This research has been partially supported by the 7th Framework
Programme for Research of the European Commission, under FET-Open
grant number 255827 (CGL Computational Geometry Learning). Partial 
support has also been provided by the Advanced Grant of the European 
Research Council GUDHI (Geometric Understanding in Higher Dimensions).
}
\\
{\normalsize INRIA, Geometrica}\\
{\normalsize Sophia-Antipolis, France}\\
\url{Jean-Daniel.Boissonnat@inria.fr}
\and
Ramsay Dyer
\footnotemark[1]
\\
{\normalsize Bernoulli Institute for Mathematics and Computer Science}\\
{\normalsize University of Groningen}\\ 
{\normalsize Groningen, The Netherlands}\\
\url{r.h.dyer@rug.nl}
\and
Arijit Ghosh
\footnote{
Supported by the Indo-German Max Planck Center for Computer Science (IMPECS).
}
\footnotemark[1]
\footnote{
Part of the work was done when the author was a visiting scientist at 
ACM Unit, Indian Statistical Institute, Kolkata, India.
}
\\
{\normalsize D1: Algorithms \& Complexity}\\
{\normalsize Max-Planck-Institut f\"ur Informatik}\\
{\normalsize Saarbr\"ucken, Germany}\\
\url{agosh@mpi-inf.mpg.de}
}

\begin{document}

\maketitle

\pagenumbering{roman}

\input{abstract}

\clearpage

\tableofcontents
\clearpage

\pagenumbering{arabic}

\input{intro}
\input{background_flt}

\input{forbidden_flt}
    \input{lem_hoop}
\input{algo}

\input{analysis}
\input{conclusions}

\phantomsection
\bibliographystyle{alpha}
\addcontentsline{toc}{section}{Bibliography}
\bibliography{delrefs}

\end{document}

%% file: environmentsRD.tex
\usepackage{amsmath,amssymb,euscript}

\usepackage{graphicx}
\usepackage{color}

\usepackage{ntheorem}

\usepackage{float}

\floatstyle{ruled}
\newfloat{structure}{Hhtb}{los}
\floatname{structure}{Structure}

\usepackage{subfig}
\usepackage{algorithmic}
\usepackage{algorithm}
\usepackage{wrapfig}

\newcounter{prob}
        {\end{list}}
\newenvironment{proof}[1][{}]{
  \begin{trivlist}\item[]\textit{Proof #1}\quad}%
  {\hfill\hspace*{\fill}~$\square$\end{trivlist}}
\newcommand{\noproof}{\hspace*{\fill}~$\square$}

\theorembodyfont{\normalfont}
\newtheorem{thm}{Theorem}[section]

\newtheorem{lem}[thm]{Lemma}
\newtheorem{de}[thm]{Definition}

\newtheorem{remk}[thm]{Remark} 

%% file: gencommands.tex
\definecolor{turquoise}{cmyk}{0.65,0,0.1,0.1}
\definecolor{deeppurple}{rgb}{0.3,0,0.3}

\newcommand{\rawdef}[1]{\emph{#1}} 
\newcommand{\defn}[1]{\rawdef{#1}\index{#1}}

\newcommand{\Algref}[1]{Algorithm~\ref{#1}}

\newcommand{\Defref}[1]{Definition~\ref{#1}}
\newcommand{\Eqnref}[1]{Equation~\eqref{#1}}
\newcommand{\Figref}[1]{Figure~\ref{#1}}
\newcommand{\Lemref}[1]{Lemma~\ref{#1}}

\newcommand{\Secref}[1]{Section~\ref{#1}}

\newcommand{\Thmref}[1]{Theorem~\ref{#1}}

\newcommand{\Remref}[1]{Remark~\ref{#1}}

\DeclareMathOperator{\convh}{conv}
\newcommand{\convhull}[1]{\convh(#1)}
\newcommand{\pwrset}[1]{2^{#1}} 

\newcommand{\bigo}[1]{O(#1)} 

\newcommand{\reel}{\mathbb{R}}

\newcommand{\rthree}{\reel^3}

\newcommand{\rem}{\reel^m}

\newcommand{\ints}{\mathbb{Z}}

\newcommand{\norm}[1]{\lVert#1\rVert}
\newcommand{\abs}[1]{\lvert#1\rvert}

\newcommand{\card}[1]{\#({#1})} 

\newcommand{\bdry}[1]{\partial{#1}}

\newcommand{\asimplex}[1]{\{#1\}} 
\newcommand{\simplex}[1]{[#1]} 
\newcommand{\carrier}[1]{\abs{#1}} 

\newcommand{\seg}[2]{\simplex{#1,#2}} 

\newcommand{\cpltcplx}[1]{\mathcal{K}(#1)}

\newcommand{\ambdim}{N}

\newcommand{\amb}{\reel^{\ambdim}} 

\newcommand{\gdist}{d} 
\newcommand{\dist}[2]{\gdist(#1,#2)}

\newcommand{\distEm}[2]{\distG{\rem}{#1}{#2}}

\newcommand{\close}[1]{\overline{#1}} 
 
\newcommand{\ball}[2]{B(#1,#2)} 
\newcommand{\cball}[2]{\close{B}(#1,#2)} 

\newcommand{\ballEm}[2]{\spaceball{\rem}{#1}{#2}} 
\newcommand{\cballEm}[2]{\cspaceball{\rem}{#1}{#2}} 

\DeclareMathOperator{\aff}{aff} 
\newcommand{\affhull}[1]{\aff(#1)}

\newcommand{\angleop}[2]{\angle(#1,#2)}

\newcommand{\pts}{\mathsf{P}}
\newcommand{\tpts}{\tilde{\pts}}

\DeclareMathOperator{\vol}{vol}

\DeclareMathOperator{\Del}{Del}
\newcommand{\delof}[1]{\Del(#1)}
\newcommand{\delP}{\delof{\pts}}

\newcommand{\pertconst}{\rho}

\newcommand{\samconst}{\epsilon}

\newcommand{\sparseconst}{\mu_0} 
\newcommand{\sparsity}{\lambda} 
 
\newcommand{\pert}{\zeta} 

\newcommand{\incl}{\iota} 
\DeclareMathOperator{\interior}{int}
\newcommand{\intr}[1]{\interior(#1)}

\newcommand{\splxs}{\sigma}

\newcommand{\splxt}{\tau}

\newcommand{\splxjoin}[2]{{#1}*{#2}}

\newcommand{\opface}[2]{#2_{#1}} 
\newcommand{\splxsp}{\opface{p}{\splxs}}
\newcommand{\splxsq}{\opface{q}{\splxs}}
\newcommand{\splxtp}{\opface{p}{\splxt}}
\newcommand{\splxtq}{\opface{q}{\splxt}}
\newcommand{\thickbnd}{\Upsilon_0}
\newcommand{\flakebnd}{\Gamma_0}
\newcommand{\thickness}[1]{\Upsilon(#1)}

\newcommand{\splxalt}[2]{D(#1,#2)} 
\newcommand{\longedge}[1]{\Delta(#1)}
\newcommand{\shortedge}[1]{L(#1)}
\newcommand{\circrad}[1]{R(#1)}
\newcommand{\circcentre}[1]{C(#1)}

\newcommand{\X}{X} 

\renewcommand{\distEm}[2]{\dist{#1}{#2}}
\renewcommand{\ballEm}[2]{\ball{#1}{#2}} 
\renewcommand{\cballEm}[2]{\cball{#1}{#2}} 

\newcommand{\torus}{{\mathbb{T}}^{m}}
\newcommand{\ballvol}[1]{V_{#1}}
\newcommand{\ballvolm}{\ballvol{m}}

\newcommand{\pertbnd}{\rho_0} 
\newcommand{\psparseconst}{\sparseconst'}
\newcommand{\psamconst}{\samconst'}
\newcommand{\ppts}{\pts'}

\newcommand{\smhullpts}{D_{\samconst}(\pts)}
\newcommand{\smhullppts}{D_{\psamconst}(\ppts)}

\newcommand{\rdelsmhull}[1]{\Del_{|}(#1)}
\newcommand{\rdelpts}{\rdelsmhull{\pts}}
\newcommand{\rdelppts}{\rdelsmhull{\ppts}}

\newcommand{\mueps}{(\sparseconst, \samconst)}
\newcommand{\pmueps}{(\psparseconst, \psamconst)}
\newcommand{\ueset}{$\mueps$-net}
\newcommand{\pueset}{$\pmueps$-net}

\newcommand{\dgconfig}{forbidden configuration} 
\newcommand{\dgconfigs}{forbidden configurations} 

\newcommand{\hoopbnd}{\alpha_0}

\newcommand{\hoop}{$\hoopbnd$-hoop}

\newcommand{\circsphere}[1]{S(#1)}
\newcommand{\diasphere}[1]{S^{m-1}(#1)}

%% file: abstract.tex
%

\begin{abstract}
  We present an algorithm that takes as input a finite point set in
  $\rem$, and performs a perturbation that guarantees that the
  Delaunay triangulation of the resulting perturbed point set has
  quantifiable stability with respect to the metric and the point
  positions. There is also a guarantee on the quality of the
  simplices: they cannot be too flat.  The algorithm provides an
  alternative tool to the weighting or refinement methods to remove
  poorly shaped simplices in Delaunay triangulations of arbitrary
  dimension, but in addition it provides a guarantee of stability for
  the resulting triangulation.
\end{abstract}

%% file: intro.tex
%

\section{Introduction}
\label{sec:intro}

The main contribution of this paper is to provide a proof that, for a
quantifiable $\delta$, a $\delta$-generic point set may be obtained as
a perturbation of an existing point set. In Euclidean space $\rem$, a
discrete point set $\pts$ is said to be \defn{$\delta$-generic} if
every Delaunay $m$-simplex has no other sample points within a
distance of $\delta$ from its circumsphere.

The Delaunay triangulation of such a point set is stable with respect
to small perturbations of either the points or of the metric
\cite{boissonnat2013stab1}. This makes $\delta$-generic sets important
in various contexts. The original motivation for this work is the
desire to establish a general framework for Delaunay triangulations on
Riemannian manifolds. 

The stability issue with geometric structures also arises in the
context of robust computation, where a high precision may be demanded
to resolve near degenerate configurations. Halperin and
Shelton~\cite{halperin1998} developed a general technique of
controlled perturbation in this setting. Funke et al.~\cite{funke2005}
presented a controlled perturbation algorithm for computing planar
Delaunay triangulations, which may be extended to higher dimensions.
Their algorithm can also be seen as seeking to produce a
$\delta$-generic point set, and in this respect, although the
motivation and context are different, our algorithm also shares some
properties with theirs.  However, in their approach all the points are
perturbed simultaneously with a probability of success that decreases
with the total size of the input point set. This makes the approach
unworkable for our desired application of triangulating general
manifolds.

By contrast, in the algorithm we present here each point is perturbed
in turn and is never subsequently visited after a successful
perturbation is found for that point. The probability of success is
independent of the total number of points or even the local
sampling density. We discuss the difference between our algorithm and
the approach of Funke et al.~\cite{funke2005} in more detail when we
conclude in \Secref{sec:conclusions}.

A well known issue with higher dimensional Delaunay triangulations is
the presence of poorly shaped (flat) ``sliver'' simplices. This
creates poorly conditioned systems in numerical applications, and
technical problems in geometric applications such as meshing
submanifolds. In fact, the issue is related to the above mentioned
problems with computing the Delaunay triangulation itself; the
existence of slivers is an indication that the point set is close to a
degenerate configuration \cite{boissonnat2013stab1}.

Existing work on removing slivers from high dimensional Euclidean
Delaunay triangulations has been based on two main techniques. The
first approach involves weighting the points to obtain a weighted
Delaunay triangulation with no slivers \cite{cheng2000}. This
technique was employed in the first work on reconstructing a
submanifold of arbitrary dimension in Euclidean space
\cite{cheng2005}, as well as in more recent work which avoids the
exponential cost of constructing a Delaunay triangulation of the
ambient space \cite{boissonnat2014tancplx.dcg}.  The other approach is to
refine the point set \cite{li2003}. This technique was used for
constructing anisotropic triangulations based on locally defined
Riemannian metrics \cite{boissonnat2011aniso.tr}, and also for meshing
submanifolds in Euclidean space \cite{boissonnat2010meshing}.  

The algorithm presented here provides a third approach, and it
guarantees a Delaunay triangulation that is stable in addition to
being sliver free. The perturbation approach enjoys the best aspects
of the other two methods. If the sample set is sufficiently dense,
there is no need to add more sample points. We also have the benefit
of using the standard metric, rather than squared distances where the
triangle inequality no longer applies. This latter aspect of the
weighting paradigm becomes awkward when considering perturbations of
the metric.

In spirit our algorithm is an extension of the algorithm presented by
Edelsbrunner et al. \cite{edelsbrunner2000smoothing} for creating a
sliver free Delaunay triangulation in $\rthree$. We extend this work
in two ways: We extend it into higher dimensions, and we also extend
it to provide $\delta$-genericity. It is this latter aspect that
embodies our primary technical contribution. In our context the
concept of sliver, and the existing extensions to higher dimensions,
were inadequate; we need to eliminate simplices that do not belong to
a Delaunay triangulation, and have no upper bound on their
circumradius.  The heart of the reason for this need to consider non-Delaunay simplices is that a violation of $\delta$-genericity is witnessed by a set $\tau$ of $m+2$ points, where $p \in \tau$ is within a distance $\delta$ of the circumsphere of the Delaunay simplex $\sigma = \tau \setminus \{p\}$. This simplex $\tau$ is not a Delaunay simplex in general, but either it, or one of its faces, represents a problem that we need to eliminate.

Our algorithm perturbs each point at most once. The correctness demonstration for this approach relies heavily on the Hoop Lemma~\ref{lem:hoop}, which says that the simplices that need to be eliminated have the property that every vertex lies close to the circumsphere of its opposing facet.

The algorithm itself is characterised by its simplicity. It is much
simpler than the refinement or weighting schemes.  In essence, at each
iteration we perturb a point $p \mapsto p'$ in such a way as to ensure
that $p'$ does not lie too close to the circumsphere of any nearby
$m$-simplex in the current point set $\ppts \setminus \{p'\}$. It is
not immediately obvious that this should result in a $\delta$-generic
point set: if $p'$ is not ``too close'' to the circumsphere of an $m$-simplex $\splxs$ in the current point set we need to be ensured that the distance from $p'$ to the circumsphere of $\splxs$ remains greater than $\delta$ even after the vertices of $\splxs$ itself have been perturbed. The analysis reveals that we can get this ensurance, even though the algorithm never explicitly considers the circumspheres of simplices containing the point that is being perturbed. 

%% file: background_flt.tex

\section{Background}
\label{sec-background-definition}

We work in $m$-dimensional Euclidean space $\rem$, where distances are
determined by the standard norm, $\norm{\cdot}$.  The distance between
a point $p$ and a set $\X \subset \rem$, is the infimum of the
distances between $p$ and the points of $\X$, and is denoted
$\distEm{p}{\X}$.  We refer to the distance between two points $a$ and
$b$ as $\norm{b-a}$ or $\distEm{a}{b}$ as convenient. A ball
$\ballEm{c}{r} = \{ x \, | \, \distEm{x}{c}< r \}$ is open, and
$\cballEm{c}{r}$ is its topological closure.  Generally, we denote the
topological closure of a set $\X$ by $\close{\X}$, the interior by
$\intr{\X}$, and the boundary by $\bdry{\X}$. The convex hull is
denoted $\convhull{\X}$, and the affine hull is $\affhull{\X}$.
The cardinality of a finite set $\pts$ is $\card{\pts}$.

\subsection{Sampling parameters}

The structures of interest will be built from a finite set $\pts
\subset \rem$, which we consider to be a set of \defn{sample
  points}. If $D \subset \rem$, then
$\pts$ is \defn{$\samconst$-dense} for $D$ if $\distEm{x}{\pts} <
\samconst$ for all $x \in D$.
We say that $\samconst$ is a \defn{sampling radius} for $D$ satisfied
by $\pts$.  If no domain $D$ is specified, we say
$\pts$ is $\samconst$-dense if $\distEm{x}{\pts \cup
  \bdry{\convhull{\pts}}} < \samconst$ for all $x \in
\convhull{\pts}$.  Equivalently, $\pts$ is $\samconst$-dense if it
satisfies a sampling radius $\samconst$ for
\begin{equation}
  \label{eq:contracted.hull}
  D_\samconst(\pts) = \{ x \in \convhull{\pts} \, | \,
  \distEm{x}{\bdry{\convhull{\pts}}} \geq \samconst \}.
\end{equation}
A convenience of this definition is expressed in
\Lemref{lem:perturb.Delone} below.

The set $\pts$ is \defn{$\sparsity$-separated} if $\distEm{p}{q} \geq
\sparsity$ for all $p,q \in \pts$. We usually assume that $\sparsity =
\sparseconst \samconst$ for some positive $\sparseconst \leq 1$. Such
a set is said to be a \defn{\ueset}, and if $\sparseconst = 1$, then
$\pts$ is an \defn{$\samconst$-net}. If $\pts$ is a \ueset\ for $D$,
then the open balls of radius $\samconst$ centred at the points of
$\pts$ cover $D$, and the likewise centred open balls of radius
$\frac{\sparseconst \samconst}{2}$ are pairwise disjoint. The sampling
radius is sometimes called a \defn{covering radius}, and
$\frac{\sparseconst \samconst}{2}$ is a \defn{packing radius} for
$\pts$.  This consistent use of open balls to describe packing and
covering radii yields the strict and non strict inequalities in our
definitions of density and separation. The density and separation parameters are used extensively in the computational geometry literature on sampling and mesh generation, while the equivalent terminology of covering radius and packing radius is favoured in the crystalography and sphere packing literature. There is no standard notation for point sets described by these parameters. In our notation $\sparseconst$ is a dimensionless quantity that gives some measure of the quality of $\pts$, while $\samconst$ is a distance and is just an indication of scale.

We work with \ueset s, but this should not be viewed as a significant constraint on the point sets considered. Indeed \emph{any} finite set of distinct points is a \ueset\ for a large enough $\samconst$ and a small enough $\sparseconst$. Thus $\samconst$ and $\sparseconst$ are simply parameters that describe the point set. However, the parameter $\sparseconst$ has a direct bearing on the output guarantees of the algorithm. Our main result, \Thmref{thm-main-theorem-of-the-paper}, reveals that the expected running time of the algorithm, as well as the stability properties of the Delaunay triangulation of the output points, both depend on $\sparseconst$.  Also, our results only begin to become interesting when $D_{\samconst}(\pts)$ defined in \Eqnref{eq:contracted.hull} is non-empty; as explained in \Secref{sec:Delaunay.complexes}, the stability claims (\Thmref{thm:thick.eucl.stability}) about Delaunay simplices only apply to simplices that are not too close to the boundary of the convex hull.

\subsection{Perturbations}

Our algorithm will return a perturbation of a given \ueset. Here we
define perturbations in our context, and observe that a perturbed
\ueset\ is itself a \pueset.
\begin{de}[Perturbation]
  \label{def:perturbation}
  A \defn{$\pertconst$-perturbation} of a \ueset\ $\pts \subset \rem$
  is a bijective application $\pert: \pts \to \ppts \subset \rem$ such
  that $\distEm{\pert(p)}{p} \leq \pertconst$ for all $p \in
  \pts$, and $\pertconst < \frac{\sparseconst \samconst}{2}$.

  For convenience, we will demand a stronger bound on $\pertconst$ and
  omit the explicit qualification: unless otherwise specified, a
  \defn{perturbation} will always refer to a
  $\pertconst$-perturbation, with $\pertconst = \pertbnd \samconst$
  for some
  \begin{equation}
    \label{eq:pertbnd}
    \pertbnd \leq \frac{\sparseconst}{4}.
  \end{equation}
  We also refer to $\ppts$ itself as a perturbation of $\pts$.
  We generally use $p'$ to denote the point $\pert(p) \in \ppts$, and
  similarly, for any point $q' \in \ppts$ we understand $q$ to be its
  preimage in $\pts$.
\end{de}

Given a perturbation constrained by \Eqnref{eq:pertbnd}, we do not expect a close relationship between the associated Delaunay complexes (defined in \Secref{sec:Delaunay.complexes}), but we can at least relate the sampling parameters of the two point sets:
\begin{lem}
  \label{lem:perturb.Delone}
  If $\pts \subset \rem$ is a \ueset, and $\ppts$ is a $\pertbnd\samconst$-perturbation of
  $\pts$, with $\pertbnd \leq \frac{\sparseconst}{4}$, then $\ppts$ is a \pueset, where
  \begin{itemize}
  \item $\psamconst = (1 + \pertbnd)\samconst \leq \frac{5}{4}\samconst$,
    and
  \item $\psparseconst = \frac{\sparseconst - 2\pertbnd}{1 + \pertbnd}
    \geq \frac{2}{5} \sparseconst$.
  \end{itemize}
\end{lem}
\begin{proof}

\begin{figure}[ht]
 \begin{center}
    \includegraphics[width=0.5\linewidth]{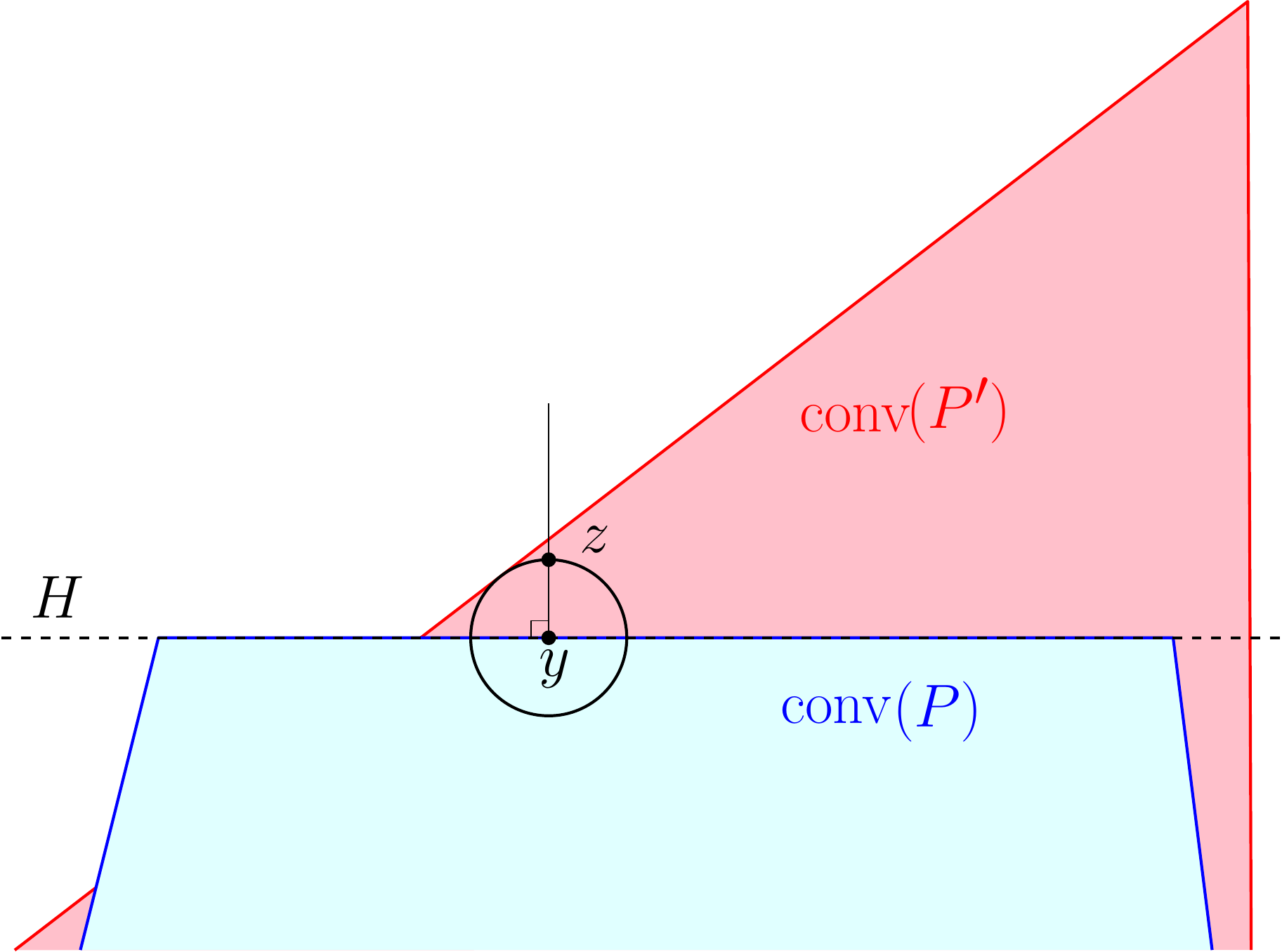} 
  \end{center}
  \caption[Perturbed convex hull]{\Lemref{lem:perturb.Delone}: $\bdry{\convhull{P}}$ and $\bdry{\convhull{P'}}$ must be close.}
  \label{fig:conv.hull.pert}
\end{figure}
The only non-trivial assertion is the density bound.
We will show that 
 \begin{equation*}
    \smhullppts \subseteq \smhullpts.
  \end{equation*}
  It follows that for any $x \in D_{\psamconst}(\ppts)$, we have
  $\distEm{x}{\ppts} \leq \distEm{x}{\pts} + \pertbnd\samconst < (1 +
  \pertbnd) \samconst = \psamconst$.

 We first observe that for any $y \in \convhull{\pts}$, we have
  \begin{equation}
    \label{eq:bnd.y.convppts}
    \distEm{y}{\convhull{\ppts}} \leq \pertbnd \samconst. 
  \end{equation}
  To see this, we use Carath\'eodory's Theorem to write $y =
  \sum_{i=0}^m \lambda_i p_i$, where $p_i \in \pts$ and the
  $\lambda_i$ are non-negative barycentric coordinates: $\sum_{i=0}^m
  \lambda_i = 1$. It follows that the point $y^* = \sum_{i=0}^m
  \lambda_i p'_i$ lies in $\convhull{\ppts}$, and $\norm{y^* - y} \leq
  \sum_{i=0}^m \lambda_i \norm{p'_i - p_i} \leq \pertbnd
  \samconst$. Similarly, we have that if $z \in \convhull{\ppts}$,
  then
  \begin{equation}
    \label{eq:bnd.z.convpts}
    \distEm{z}{\convhull{\pts}} \leq \pertbnd \samconst. 
  \end{equation}

  This implies that if $y \in \bdry{\convhull{\pts}}$, then
  $\distEm{y}{\bdry{\convhull{\ppts}}} \leq \pertbnd
  \samconst$. Indeed, assume that $y \in \convhull{\ppts}$, since
  otherwise the assertion is an immediate consequence of
  \Eqnref{eq:bnd.y.convppts}. To reach a contradiction, assume
  $\distEm{y}{\bdry{\convhull{\ppts}}} = R > \pertbnd \samconst$. Then
  $\close{B} = \cballEm{y}{R} \subseteq \convhull{\ppts}$. Let $H$ be
  a hyperplane through $y$ and supporting $\convhull{\pts}$, and let
  $z \in \bdry{\close{B}}$ lie on a line through $y$ and orthogonal to
  $H$ and in the open half-space that doesn't contain
  $\convhull{\pts}$, as shown in \Figref{fig:conv.hull.pert}. Then
  $\dist{z}{\convhull{\pts}} = R > \pertbnd \samconst$, contradicting
  \Eqnref{eq:bnd.z.convpts}.

  Suppose $x \in \smhullppts$. Let $y \in
  \bdry{\convhull{\pts}}$ be such that $\distEm{x}{y} =
  \distEm{x}{\bdry{\convhull{\pts}}}$, and let $z \in
  \bdry{\convhull{\ppts}}$ satisfy $\distEm{y}{z} =
  \distEm{y}{\bdry{\convhull{\ppts}}}$. Then
  \begin{equation*}
    \begin{split}
      \psamconst &\leq \distEm{x}{z} \leq \distEm{x}{y} +
      \distEm{y}{z}\\
      &= \distEm{x}{\bdry{\convhull{\pts}}} +
      \distEm{y}{\bdry{\convhull{\ppts}}}\\
      &\leq \distEm{x}{\bdry{\convhull{\pts}}} + \pertbnd \samconst,
    \end{split}
 \end{equation*}
 and we obtain $\distEm{x}{\bdry{\convhull{\pts}}} \geq \psamconst -
 \pertbnd\samconst = \samconst$. Hence $x \in \smhullpts$.
\end{proof}

%

\subsection{Simplices}

Although our problem setting is geometric in nature, it is convenient
to work with the framework of abstract simplices and complexes.  A
\defn{simplex} $\splxs$ is a non-empty finite set. The
\defn{dimension} of $\splxs$ is given by $\dim{\splxs} = \card{\splxs}
- 1$, and a $j$-simplex refers to a simplex of dimension $j$. The
dimension of a simplex is sometimes indicated with a superscript:
$\splxs^j$.  The elements of $\splxs$ are called the \defn{vertices}
of $\splxs$. We do not distinguish between a $0$-simplex and its
vertex. If a simplex $\splxs$ is a subset of $\splxt$, we say it is a
\defn{face} of $\splxt$, and we write $\splxs \leq \splxt$. A
$1$-dimensional face is called an \defn{edge}. If $\splxs$ is a proper
subset of $\splxt$, we say it is a \defn{proper face} and we write
$\splxs < \splxt$.  A \defn{facet} of $\splxt$ is a face $\splxs$ with
$\dim{\splxs} = \dim{\splxt} - 1$.

For any vertex $p \in \splxs$, the \defn{face opposite} $p$ is the
face determined by the other vertices of $\splxs$, and is denoted
$\splxsp$. If $\splxs$ is a $j$-simplex, and $p$ is not a vertex of
$\splxs$, we may construct a $(j+1)$-simplex $\splxt =
\splxjoin{p}{\splxs}$, called the \defn{join} of $p$ and $\splxs$. It
is the simplex defined by $p$ and the vertices of $\splxs$, i.e.,
$\splxs = \splxtp$.

We will be considering simplices whose vertices are points in $\rem$,
and this endows the simplices with geometric properties, but we do not
require the vertices to be affinely independent. If $\splxs \subset
\rem$ and $x \in \splxs$, then $x$ is a vertex of $\splxs$.

The \defn{length} of an edge is the distance between its vertices. The
\defn{diameter} of a simplex $\splxs$ is its longest edge length, and
is denoted $\longedge{\splxs}$. The shortest edge length is
denoted $\shortedge{\splxs}$.  If $\splxs$ is a $0$-simplex, we define
$\shortedge{\splxs} = \longedge{\splxs} = 0$.

The \defn{altitude} of $p$ in $\splxs$ is $\splxalt{p}{\splxs} =
\distEm{p}{\affhull{\splxsp}}$. A poorly-shaped simplex can be
characterized by the existence of a relatively small altitude. The
\defn{thickness} of a $j$-simplex $\splxs$ is the dimensionless
quantity
\begin{equation*}
  \thickness{\splxs} =
  \begin{cases}
    1& \text{if $j=0$} \\
    \min_{p \in \splxs} \frac{\splxalt{p}{\splxs}}{j
      \longedge{\splxs}}& \text{otherwise.}
  \end{cases}
\end{equation*}
We say that $\splxs$ is $\thickbnd$-thick, if $\thickness{\splxs} \geq
\thickbnd$. If $\splxs$ is $\thickbnd$-thick, then so are
all of its faces. Indeed if $\splxs^j \leq \splxs$, then the smallest
altitude in $\splxs^j$ cannot be smaller than that of $\splxs$, and also
$\longedge{\splxs^j} \leq \longedge{\splxs}$. 

A \defn{circumscribing ball} for a simplex $\splxs$ is any
$m$-dimensional ball that contains the vertices of $\splxs$ on its
boundary.  If $\thickness{\splxs} = 0$, we say that $\splxs$ is
\defn{degenerate}, and such a simplex may not admit any circumscribing
ball.  If $\splxs$ admits a circumscribing ball, then it has a
\defn{circumcentre}, $\circcentre{\splxs}$, which is the centre of the
unique smallest circumscribing ball for $\splxs$. The radius of this
ball is the \defn{circumradius} of $\splxs$, denoted
$\circrad{\splxs}$. A degenerate simplex $\splxs$ may or may not have
a circumcentre and circumradius; we write $\circrad{\splxs} < \infty$
to indicate that it does.  In this case we can also define the
\defn{diametric sphere} as the boundary of the smallest circumscribing
ball: $\diasphere{\splxs} =
\bdry{\ballEm{\circcentre{\splxs}}{\circrad{\splxs}}}$, and the
\defn{circumsphere}: $\circsphere{\splxs} = \diasphere{\splxs} \cap
\affhull{\splxs}$. Observe that if $\splxs \leq \splxt$, then
$\circsphere{\splxs} \subseteq \circsphere{\splxt}$. If $\dim \splxs =
m$, then $\circsphere{\splxs} = \diasphere{\splxs}$.

\subsection{Complexes}

An \defn{abstract simplicial complex} (we will just say
\defn{complex}) is a set $\mathcal{K}$ of simplices such that if
$\splxs \in \mathcal{K}$, then all the faces of $\splxs$ are also
members of $\mathcal{K}$. The union of the vertices of all the
simplices of $\mathcal{K}$ is the \defn{vertex set} of
$\mathcal{K}$. We say that $\mathcal{K}$ is a \defn{complex on $\pts$}
if $\pts$ includes the vertex set of $\mathcal{K}$. Our complexes are
finite and the number of simplices in a complex $\mathcal{K}$ is
denoted $\card{\mathcal{K}}$.  The \defn{complete complex} on $\pts$,
denoted $\cpltcplx{\pts}$, is set of all simplices that have vertices
in $\pts$. If we let $\pwrset{\pts}$ denote the set of subsets of
$\pts$, then $\cpltcplx{\pts} = \pwrset{\pts} \setminus \emptyset$. A
complex $\mathcal{K}$ is the complete complex on $\pts$ if and only if
$\pts$ is the vertex set of $\mathcal{K}$ and $\pts \in \mathcal{K}$.

A subset $\mathcal{L} \subseteq \mathcal{K}$ is a \defn{subcomplex} of
$\mathcal{K}$ if it is also a complex. If $\mathcal{K}$ is a complex
on $\pts$, and $\mathcal{K}'$ is a complex on $\ppts$, then a map
$\pert: \pts \to \ppts$ induces a \defn{simplicial map} $\mathcal{K}
\to \mathcal{K}'$ if for every $\splxs \in \mathcal{K}$,
$\pert(\splxs) \in \mathcal{K}'$. Thus the image of the simplicial map is a
subcomplex of $\mathcal{K}'$. We denote the simplicial map with
the same symbol, $\pert$. If $\pert$ is injective on
$\pts$, and $\pert(\mathcal{K}) = \mathcal{K}'$, then $\pert$
is an \defn{isomorphism}.

Although we prefer to work with abstract simplices and complexes, the
underlying motivation for this work is centred in the concept of a
\defn{triangulation}, which demands traditional geometric simplicial
complexes for its definition. A \defn{geometric realisation} of a
complex $\mathcal{K}$ with vertex set $\pts$, is a topological
space $\carrier{\mathcal{K}} \subset \amb$ such that there is a
bijection $g: \pts \to \tpts \subset \carrier{\mathcal{K}}$ with the
property that $\bigcup_{\splxs \in \mathcal{K}} \convhull{g(\splxs)} =
\carrier{\mathcal{K}}$, and if $\splxt, \splxt' \in \mathcal{K}$, then
$\convhull{g(\splxt)} \cap \convhull{g(\splxt')} = \X$, where either
$\X = \emptyset$, or $\X = \convhull{g(\splxs)}$
with $\splxs = (\splxt \cap \splxt') \in \mathcal{K}$.

If $\mathcal{K}$ is a complex on $\pts \subset \rem$, we say that
$\mathcal{K}$ is \defn{embedded} if the inclusion map $\incl: \pts
\hookrightarrow \rem$ yields a geometric realisation of
$\mathcal{K}$. A \defn{triangulation of a connected set $\X \subset
  \rem$} is an embedded complex $\mathcal{K}$ on $\pts \subset \X$
such that $\carrier{\mathcal{K}} = \X$. A \defn{triangulation of
$\pts \subset \rem$} is a triangulation of $\convhull{\pts}$.

\subsection{Delaunay complexes}
\label{sec:Delaunay.complexes}

Our definition of the Delaunay complex is equivalent to defining it as
the nerve of the Voronoi diagram, however we do not exploit
the Voronoi diagram in this work. 

An \defn{empty ball} is one that contains no point from $\pts$. 
\begin{de}[Delaunay complex]
  \label{def:Delaunay.complex}
  A \defn{Delaunay ball} is a maximal empty ball. Specifically, $B =
  \ballEm{x}{r}$ is a Delaunay ball if any empty ball centred at $x$
  is contained in $B$. A simplex $\splxs$ is a \defn{Delaunay
    simplex} if there exists some Delaunay ball $B$ such that the
  vertices of $\splxs$ belong to $\bdry{B}\cap \pts$.  The
  \defn{Delaunay complex} is the set of Delaunay simplices, and is
  denoted $\delP$.
\end{de}

If $\X \subset \rem$, then the \defn{Delaunay complex of $\pts$
  restricted to $\X$} is the subcomplex of $\delP$ consisting of those
simplices that have a Delaunay ball centred in $\X$. We are interested
in the case where $\X = \smhullpts$ for a finite $\samconst$-dense
sample set $\pts$. We denote the Delaunay complex of $\pts$ restricted
to $\smhullpts$ by $\rdelpts$. Our interest in this subcomplex is due
to the following observation that is an immediate consequence of the
definitions. If the radius of a Delaunay ball $\splxs$ exceeds $\samconst$, then the centre of that ball is at a distance of more than $\samconst$ from any point in $\pts$. Thus we have:
\begin{lem}
  \label{lem:rdel.small.circrad}
  If $\pts$ is $\samconst$-dense, then
  every simplex $\splxs \in \rdelpts$ has a Delaunay ball with radius
  less than $\samconst$, and in particular $\circrad{\splxs} <
  \samconst$. 
\end{lem}

A Delaunay simplex $\splxs$ is \defn{$\delta$-protected} if it has a
Delaunay ball $B$ such that $\distEm{q}{\bdry{B}} > \delta$ for all
$q \in \pts \setminus \splxs$.  We say that $B$ is a
$\delta$-protected Delaunay ball for $\splxs$. 
We say that $\splxs$ is \defn{protected} to mean that it is
$\delta$-protected for some unspecified $\delta > 0$.

A \ueset\ $\pts \subset \rem$ is \defn{$\delta$-generic} if all the Delaunay $m$-simplices in $\rdelpts$ are $\delta$-pro\-tec\-ted. The set $\pts$ is simply \defn{generic} if it is $\delta$-generic for some unspecified $\delta > 0$.  If $\pts$ is generic, then $\rdelpts$ is embedded \cite[Lemmas 3.5]{boissonnat2013stab1}, and with an abuse of language we call $\rdelpts$ the \defn{restricted Delaunay triangulation} of $\pts$. (We are abusing the language because in general $\rdelpts$ coincides with neither $\convhull{P}$ nor $D_{\samconst}$.)
If $\pts$ is a $\delta$-generic \ueset, then the Delaunay
triangulation exhibits stability with respect to small perturbations
of the points or of the metric~\cite{boissonnat2013stab1}. This gives us motivation to demonstrate that $\delta$-generic point sets can be produced algorithmically, which is the primary contribution of the current work. 

We will present an algorithm that, when given a \ueset, and a small
positive parameter $\flakebnd < 1$, will generate a $\delta$-generic
\pueset\ $\ppts$ such that all the $m$ simplices in $\rdelppts$ are
$\flakebnd^m$-thick.  As an example in this context, the stability
with respect to the sample positions \cite[Theorem
4.14]{boissonnat2013stab1}, can be stated as:
\begin{thm}[Delaunay stability]
  \label{thm:thick.eucl.stability}
  Suppose $\ppts \subset \rem$ is a \pueset, and all the $m$-simplices
  in $\rdelppts$ are $\flakebnd^m$-thick and $\delta$-protected, where
  $\delta = \delta_0 \psparseconst \psamconst$, with $0 \leq \delta_0
  \leq 1$. If $\pert: \ppts \to \tpts$ is a $\pertconst$-perturbation
  of $\ppts$ with
  \begin{equation*}
    \pertconst \leq \frac{\flakebnd^m \psparseconst^2
      \delta_0}{18}\psamconst, 
  \end{equation*}
  then $\pert: \rdelppts \to \mathcal{K} \subseteq \delof{\tpts}$ is a
  simplicial isomorphism onto an embedded subcomplex $\mathcal{K}$ of
  $\delof{\tpts}$. 
\end{thm}

%% file: forbidden_flt.tex

\section{Forbidden configurations}
\label{sec:forbidden}

Our goal is to produce a point set whose Delaunay triangulation has
nice properties.  In this section we identify specific configurations
of points whose existence in a \pueset\ $\ppts$ implies that $\ppts$
does not meet the requirements of \Thmref{thm:thick.eucl.stability}.
These configurations are a particular family of thin simplices that we
call \defn{\dgconfig s}.

For a \ueset\ the Delaunay triangles automatically enjoy a lower bound
on their thickness due to the bounds on their circumradius and
shortest edge (as verified by a calculation similar to the one in Lemma 3.13 of the Delaunay stability paper~\cite{boissonnat2013stab1}). However, higher dimensional Delaunay simplices may have
arbitrarily small thickness. The problem simplices in three
dimensional Delaunay triangulations have their vertices all near ``the
equator'' of their circumsphere, and were dubbed
\defn{slivers}~\cite{cheng2000}. They were characterised as simplices
that had an upper bound on both their thickness and the ratio of their
circumradius to shortest edge length.

The essential property of slivers, that is exploited by many
algorithms that seek to remove them, is the fact that every vertex
lies close to the circumcircle of its opposing facet. This property is
a consequence of the defining characteristics of a sliver, and it is
demonstrated in a ``Torus Lemma'' \cite{edelsbrunner2000smoothing}.
The Torus Lemma is important because it places a bound on the volume
of possible positions of a fourth vertex that would make a sliver when
joined with a fixed set of three vertices.  

The concept of a sliver has been extended to higher dimensions in
various works, and likewise there is a higher dimensional analogue of
the Torus Lemma \cite{li2003}.  In our current context, we will be
considering unwanted simplices that are not subjected to an upper
bound on their circumradius, because they are not Delaunay
simplices. For this reason, we introduce \defn{flakes} in
\Secref{sec:flakes}. Flakes have one of the important properties of
slivers: there is an upper bound on all of the altitudes, but flakes
are not subjected to a circumradius bound.

A flake that appears in the Delaunay complex of a \ueset\ is
necessarily a sliver in the traditional sense, but the Torus Lemma
does not apply to flakes in general.  In \Secref{sec:delta.gen.forbid}
we introduce the \defn{\dgconfig s}, a subfamily of flakes that may be
considered to be a generalisation of slivers.  In
\Secref{sec:hoop.property} we show that \dgconfig s will exhibit
the important property embodied in the Torus Lemma. We call this
property the \defn{hoop property}, and the Hoop \Lemref{lem:hoop} is
our extension of the Torus Lemma to the current context.


\subsection{Flakes}
\label{sec:flakes}

In dimensions higher than three, a simple upper bound on the thickness
of a simplex is not sufficient to bound \emph{all} of the altitudes of
the simplex. In order to obtain an effective bound on all of the
altitudes, a small upper bound on the thickness needs to be coupled
with a relatively larger lower bound on the thickness of the
facets. For this reason we introduce a thickness requirement that is
gradated with the dimension. We exploit a positive real parameter
$\flakebnd$, which is no larger than one. In the following definition, $\flakebnd^j$ means $\flakebnd$ raised to the $j^{\text{th}}$ power.
\begin{de}[$\flakebnd$-good simplices and $\flakebnd$-flakes]
  \label{def:flake}
  \label{def:good.simplex}
  A simplex $\splxs$ is \defn{$\flakebnd$-good} if for all $j$ with
  $0\leq j \leq \dim \splxs$, we have $\thickness{\splxs^j} \geq
  \flakebnd^j$ for all $j$-simplices $\splxs^j \leq \splxs$. A simplex
  is \defn{$\flakebnd$-bad} if it is not $\flakebnd$-good. A
  \defn{$\flakebnd$-flake} is a $\flakebnd$-bad simplex in which all
  the proper faces are $\flakebnd$-good.
\end{de}
Observe that a flake must have dimension at least $2$, since
$\thickness{\splxs^j} = 1$ for $j < 2$. Also, since a flake may be
degenerate, but its facets cannot, the dimension of a flake can be as
high as $m+1$, but no higher.

Earlier definitions of slivers  in higher dimensions
\cite{li2003,cheng2005} correspond to flakes together with the
additional requirement that the circumradius to shortest edge ratio be
bounded.  The dimension-gradated requirement on simplex quality
(altitude bound) is implicitly present in these earlier works.

Ensuring that all simplices in a complex $\mathcal{K}$ are
$\flakebnd$-good is the same as ensuring that there are no flakes in
$\mathcal{K}$.  Indeed, if $\splxs$ is $\flakebnd$-bad, then it has a
$j$-face $\splxs^j \leq \splxs$ that is not $\flakebnd^j$-thick. By
considering such a face with minimal dimension we arrive at the
following important observation:
\begin{lem}
  \label{lem:bad.has.flake}
  A simplex is $\flakebnd$-bad if and only if it has a face that is a
  $\flakebnd$-flake.
\end{lem}

We obtain an upper bound on the altitudes of a $\flakebnd$-flake
through a consideration of dihedral angles. In particular, we observe
the following general relationship between simplex altitudes:
\newcommand{\spq}{\splxs_{pq}}
\begin{figure}[ht]
  \begin{center}
    \includegraphics[width=0.5\linewidth]{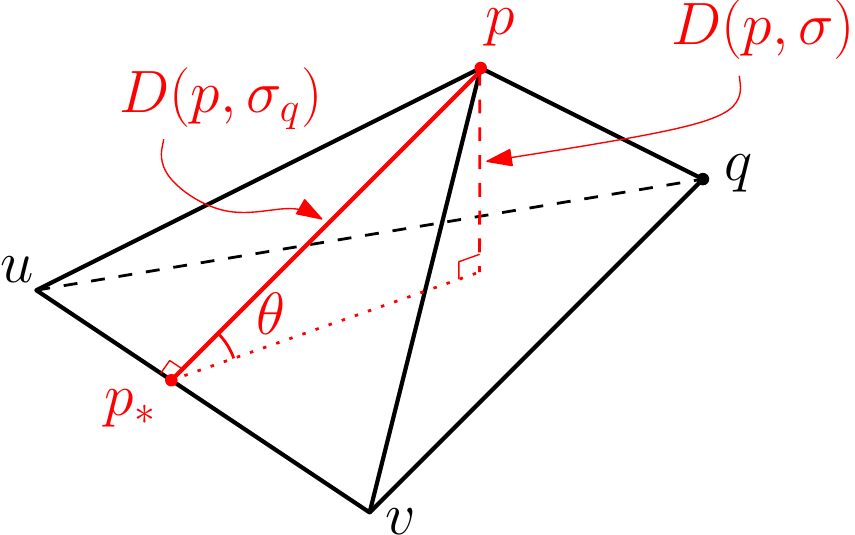} 
  \end{center}
  \caption[Dihedral angle]{The sine of the dihedral angle $\theta$ between the facets $\splxsq = \asimplex{p,u,v}$, and $\splxsp = \asimplex{q,u,v}$ of $\splxs = \asimplex{p,q,u,v}$ is given by $\frac{\splxalt{p}{\splxs}}{\splxalt{p}{\splxsq}}$, i.e., the ratio of the altitude of $p$ in $\splxs$ to the altitude of $p$ in $\splxsq$. The point $p_*$ is the orthogonal projection of $p$ into the affine hull of $\spq = \asimplex{u,v}$.
 }
  \label{fig:dihedral}
\end{figure}
\begin{lem}
  \label{lem:alt.ratios}
  If $\splxs$ is a $j$-simplex with $j \geq 2$, then for any two
  vertices $p,q \in \splxs$, the dihedral angle between $\splxsp$ and
  $\splxsq$ defines an equality between ratios of altitudes:
  \begin{equation*}
    \sin \angleop{\affhull{\splxsp}}{\affhull{\splxsq}} =
    \frac{\splxalt{p}{\splxs}}{\splxalt{p}{\splxsq}}
    =
    \frac{\splxalt{q}{\splxs}}{\splxalt{q}{\splxsp}}.
  \end{equation*}
\end{lem}
\begin{proof}
  An example of the assertion is depicted in \Figref{fig:dihedral}.
  Let $\spq = \splxsp \cap \splxsq$, and let $p_*$ be the projection
  of $p$ into $\affhull{\spq}$. Taking $p_*$ as the origin, we see
  that $\frac{p-p_*}{\splxalt{p}{\splxsq}}$ has the maximal distance
  to $\affhull{\splxsp}$ out of all the unit vectors in
  $\affhull{\splxsq}$, and this distance is
  $\frac{\splxalt{p}{\splxs}}{\splxalt{p}{\splxsq}}$. By definition
  this is the sine of the angle between $\affhull{\splxsp}$ and
  $\affhull{\splxsq}$. A symmetric argument is carried out with $q$ to
  obtain the result.
\end{proof}

The usefulness of the definition of flakes lies in the following observation:
\begin{lem}[Flakes have small altitude]
  \label{lem:flake.alt.bnd}
  If $\splxt$ is a $\flakebnd$-flake, then for any vertex $p \in
  \splxt$,
  \begin{equation*}
    \splxalt{p}{\splxt} < \frac{ 2\longedge{\splxt}^{2}
      \flakebnd }{ \shortedge{\splxt} }.
  \end{equation*}
\end{lem}
\begin{proof}
  Recalling \Lemref{lem:alt.ratios} we have
  \begin{equation*}
    \splxalt{p}{\splxt} = \frac{\splxalt{q}{\splxt}
      \splxalt{p}{\splxtq} } {\splxalt{q}{\splxtp}},
  \end{equation*}
  and taking $q$ to be a vertex with minimal altitude, we have
  \begin{equation*}
    \splxalt{q}{\splxt} = k \thickness{\splxt} \longedge{\splxt}
    < k \flakebnd^k \longedge{\splxt},
  \end{equation*}
  and
  \begin{equation*}
    \begin{split}
      \splxalt{q}{\splxtp}
      &\geq (k-1) \thickness{\splxtp} \longedge{\splxtp}\\
      &\geq (k-1) \flakebnd^{k-1} \longedge{\splxtp}\\
      &\geq (k-1) \flakebnd^{k-1} \shortedge{\splxt},
    \end{split}
  \end{equation*}
  and
  \begin{equation*}
    \splxalt{p}{\splxtq} \leq \longedge{\splxtq} \leq \longedge{\splxt},
  \end{equation*}
  and since $k \leq 2(k-1)$, the bound is obtained.
\end{proof}

\subsection{Properties of $\delta$-generic point sets}
\label{sec:delta.gen.forbid}

\begin{figure}[ht]
  \begin{center}
    \includegraphics[width=0.5\linewidth]{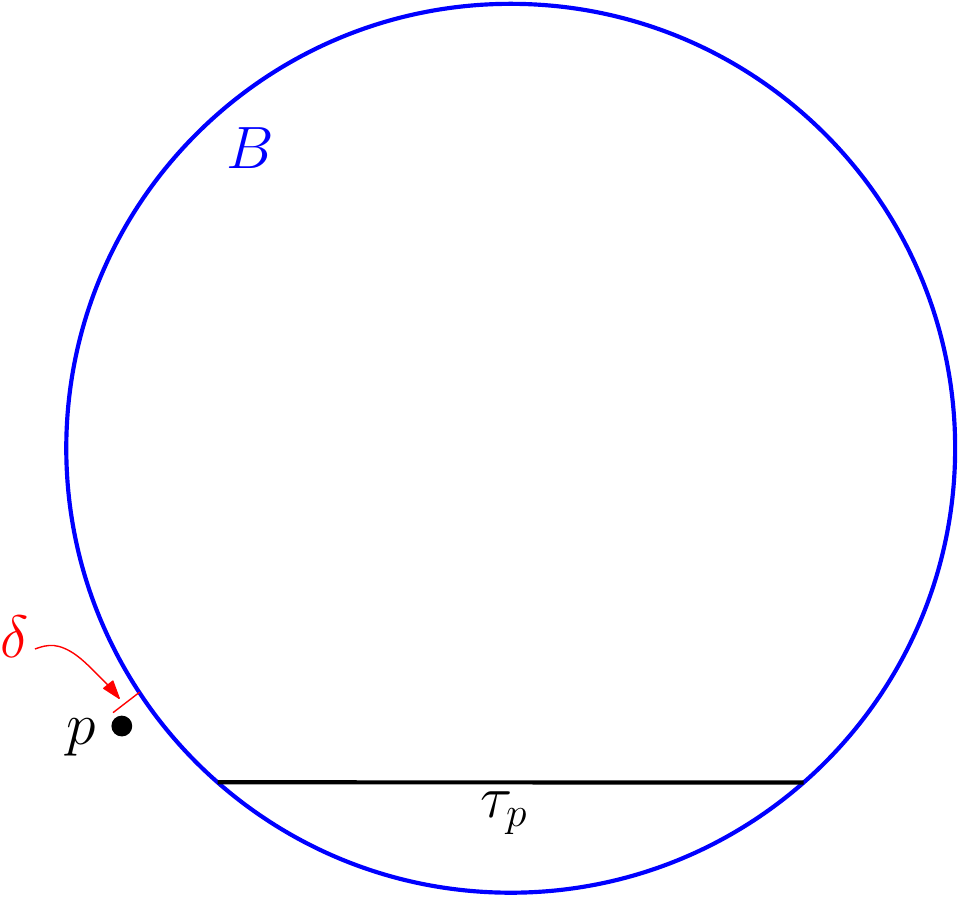} 
  \end{center}
  \caption[Forbidden configuration]{ A forbidden configuration is a
    flake $\splxt$ that has a vertex $p$ that lies within a distance
    $\delta$ from a small circumscribing ball of the opposing facet
    $\splxtp$.}
  \label{fig:forbidden.config}
\end{figure}
In order to ensure a $\delta$-generic point set $\ppts$, we need to
consider simplices that may not appear in any Delaunay
triangulation. Specifically, we do not have a circumradius bound on
the problem configurations. This makes their description more
complicated than the traditional definition of a sliver.  As
schematically depicted in \Figref{fig:forbidden.config}, we have the
following characterisation of the configurations that we need to
avoid:
\begin{de}[Forbidden configuration]
  \label{def:forbidden.config}
  Let $\pts' \subset \rem$ be a \pueset.  A $(k+1)$-simplex $\splxt
  \subseteq \pts'$, is a \defn{\dgconfig} in $\ppts$ if it is a
  $\flakebnd$-flake, with $k\leq m$, and there exists a $p \in \splxt$
  such that $\opface{p}{\splxt}$ has a circumscribing ball $B =
  \ballEm{C}{R}$ with $R < \psamconst$, and $\abs{\distEm{p}{C} - R}
  \leq \delta$, where $\delta = \delta_0 \psparseconst \psamconst$. We
  say that the \dgconfig\ is \defn{certified} by $p$ and $B$.
\end{de}
We remark that the definition of a \dgconfig\ depends on two
parameters, $\flakebnd$, and $\delta_0$, as well as on the parameters
which we associate with the sample set $\ppts$, namely
$\psparseconst$, and $\psamconst$.

In order to guarantee that the \pueset\ $\ppts$ is $\delta$-generic,
with $\delta = \delta_0 \psparseconst \psamconst$, it is sufficient to
ensure that there is no \dgconfig\ with vertices in $\ppts$:
\begin{lem}
  \label{lem:non.protection.implies.forbidden}
  Suppose $\ppts \subset \rem$ is a \pueset. If there exists an
  $m$-simplex $\splxs^m \in \rdelppts$ which is not
  $\delta$-protected, with $\delta = \delta_0 \psparseconst
  \psamconst$, then $\cpltcplx{\ppts}$ contains a \dgconfig. Likewise,
  if any $\splxs^m \in \rdelppts$ is not $\flakebnd$-good, then
  $\cpltcplx{\ppts}$ contains a \dgconfig.
\end{lem}
\begin{proof}
  Suppose $\splxs^m \in \rdelppts$ is not $\delta$ protected. Then
  there exists a $p \in \ppts \setminus \splxs^m$ such that $0 \leq
  \distEm{p}{\circcentre{\splxs^m}} - \circrad{\splxs^m} \leq \delta$. The
  $(m+1)$-simplex $\tilde{\splxt} = \splxjoin{p}{\splxs^m}$ is
  necessarily degenerate, therefore, by \Lemref{lem:bad.has.flake},
  there is a $\flakebnd$-flake $\splxt \leq \tilde{\splxt}$. If $p$ belongs
  to $\splxt$, then $\splxt$ is necessarily a \dgconfig\ certified by
  $p$ and $B = \ballEm{\circcentre{\splxs^m}}{\circrad{\splxs^m}}$,
  because $\delta \leq \delta_0 \psparseconst \psamconst$. If $p$ does
  not belong to $\splxt$, then it is a \dgconfig\ certified by any one
  of its vertices and $B$.

  A similar argument reveals a \dgconfig\ if $\splxs^m$ is not
  $\flakebnd$-good. 
\end{proof}

\subsection{The Hoop property}
\label{sec:hoop.property}

We characterise the property of \dgconfig s that is important for
algorithmic purposes as follows:
\begin{de}[Hoop property]
  \label{def:hoop}
  A simplex $\splxt \subset \rem$ has the \defn{\hoop\ property} if
  there is a constant $\hoopbnd > 0$ such that for every $p \in
  \splxt$, the opposing facet has a circumcentre and
  \begin{equation*}
    \distEm{p}{\circsphere{\opface{p}{\splxt}}} \leq \hoopbnd
    \circrad{\opface{p}{\splxt}} < \infty.
  \end{equation*}
\end{de}

\subsubsection{The Hoop Lemma}
\label{sec:hoop.lem}

We emphasise that the \emph{symmetric} nature of the hoop property is essential for our purposes. The hoop property says that \emph{every} vertex is close to the circumsphere of the opposing facet. We obtain this bound in two steps. First we exploit the thickness of the facets to show that \dgconfig s have a natural symmetry characterised by the fact that \emph{every} vertex lies close to some small circumscribing sphere of its opposing facet:
\begin{lem}[Symmetry of \dgconfigs]
  \label{lem:close.to.some.sphere}
  Suppose $\splxt = \splxjoin{q}{\splxs}$ is a $(k+1)$-simplex
  certified by $q$ and $\ballEm{C}{R}$ as a \dgconfig\ in a
  \pueset. If $\delta_0 \leq \frac{1}{4}$, then for any $p \in
  \splxt$ there exists a ball $B = \ballEm{C_p}{R_p}$ circumscribing
  $\splxtp$ and such that
  \begin{equation*}
    R_p \leq \left(1 + \frac{3\delta_0}{\psparseconst \flakebnd^k}
    \right) R,
  \end{equation*}
  and
  \begin{equation*}
    \distEm{p}{\bdry{B}} \leq \left( \frac{6\delta_0}{\psparseconst^2
        \flakebnd^k} \right) \shortedge{\splxtp}.
  \end{equation*}
\end{lem}
\begin{proof}
  The idea is that $C$ is ``almost'' a circumcentre for $\splxtp$ in that the distances between $C$ and the vertices of $\splxtp$ are all very close. Since $\splxtp$ is thick, we can exploit a result \cite[Lemma 4.3]{boissonnat2013stab1} that says that $\splxtp$ must have a circumscribing ball with a centre near $C$. The bounds then follow from a consideration of the triangle inequality, and the fact that $\splxtp$ and $\splxs$ must have a vertex in common.

  We observe that for any $u,v \in \splxtp$ we have
  \begin{equation*}
    \big| \distEm{u}{C} - \distEm{v}{C} \big| \leq \delta_0
    \shortedge{\splxs}. 
  \end{equation*}
  It follows then, from \cite[Lemma 4.3]{boissonnat2013stab1}, that
  there is a circumscribing ball $B = \ballEm{C_p}{R_p}$ for $\splxtp$
  with
  \begin{equation*}
    \distEm{C_p}{C} \leq \frac{ (R + \delta_0
      \shortedge{\splxs} ) \delta_0 \shortedge{\splxs}}
    {\thickness{\splxtp}\longedge{\splxtp} }.
  \end{equation*}
  Since $\splxt$ is a $\flakebnd$-flake, $\thickness{\splxtp} \geq
  \flakebnd^k$. Thus, using
  $\frac{\shortedge{\splxs}}{\longedge{\splxtp}} \leq
  \frac{2}{\psparseconst}$, and $R \leq
  \frac{1}{\psparseconst}\shortedge{\splxtp}$ and $\delta_0 <
  \frac{1}{4}$, we find
  \begin{equation*}
    \distEm{C_p}{C} \leq \frac{ 2(R + 2\delta_0 R) \delta_0}
    {\psparseconst \flakebnd^k}
    \leq \frac{ 3\delta_0 R}
    {\psparseconst \flakebnd^k} \leq 
    \frac{ 3 \delta_0 \shortedge{\splxtp} } 
      {\psparseconst^2 \flakebnd^k}.
  \end{equation*}
  We have $k \geq 1$, since $\splxt$ is a flake, so $\splxs$ and $\splxtp$
  must share a common vertex. Thus the bounds follow from the triangle
  inequality. 
\end{proof}

In the next step we arrive at the \hoop\ property by exploiting the
altitude bound on every vertex that is guaranteed by
\Lemref{lem:flake.alt.bnd} because a \dgconfig\ is a
$\flakebnd$-flake.  The Symmetry \Lemref{lem:close.to.some.sphere}
allows us to exploit an argument similar to the traditional
demonstration of the torus lemma. The full proof is described in
\Secref{sec:hoop.lem.proof}. We arrive at the following Hoop Lemma, which is
a restatement of \Lemref{lem:shell.lem}:
\begin{lem}[Hoop Lemma]
  \label{lem:hoop}
  If
  \begin{equation*}
    \delta_0 \leq \frac{\psparseconst^2 \flakebnd^m}{6},
  \end{equation*}
  then a \dgconfig\ $\splxt$ in a \pueset\ has the \hoop\ property
  with
  \begin{equation*}
    \hoopbnd = \left(\frac{6}{\psparseconst} \right)^3
    \left( \flakebnd + \frac{\delta_{0}}{\flakebnd^m} \right).
  \end{equation*}
  Furthermore,
  the facets of $\splxt$ are subject to a circumradius bound:
  \begin{equation*}
    \circrad{\splxtp} < \left(1 + \frac{3 \delta_0}{\psparseconst
        \flakebnd^m} \right) \psamconst,
 \end{equation*}
 for all $p \in \splxt$. 
\end{lem}
The definition of \dgconfig s is cumbersome, but the Hoop
\Lemref{lem:hoop} provides us with a symmetric property of \dgconfig s
that is easy to exploit. In particular, when we perturb a point $p \mapsto p'$, then for any nearby simplex $\splxs$, we are able to check whether $\splxt = \splxjoin{p'}{\splxs}$ is a forbidden configuration simply by examining the distance between $p'$, and the circumsphere for $\splxs$; we do not have to check this for all the vertices of $\splxt$.

\subsubsection{The perturbation setting}

Although we have described \dgconfig s and the Hoop Lemma in terms of
a \pueset\ $\ppts$, rather than a \ueset\ $\pts$, the notation is
simply a convenience for our current purposes. Until now we have not
supposed that $\ppts$ was a perturbation of a \ueset. We now review
the results in this setting.

If we constrain $\flakebnd$ and constrain $\delta_0$ relative to
$\flakebnd$, we observe that, for a \dgconfig\ that appears in a
perturbed point set, the properties expressed in the Hoop
\Lemref{lem:hoop}
can be simplified and, by using \Lemref{lem:perturb.Delone}, they can
be expressed in terms of the parameters of the original \ueset:
\begin{lem}[Hoop Lemma for perturbed points]
  \label{lem:clean.bad.hoop}
  Suppose $\ppts$ is a perturbation of the \ueset\ $\pts$, and $\splxt
  \subset \ppts$ is a \dgconfig.
  If
  \begin{equation*}
    \delta_0 \leq \flakebnd^{m+1} \quad \text{and} \quad \flakebnd
    \leq \frac{2\sparseconst^2 }{75},    
 \end{equation*}
  then $\splxt$ has the \hoop\ property, with
  \begin{equation*}
    \hoopbnd = 2 \left( \frac{16}{\sparseconst} \right)^3 \flakebnd.
  \end{equation*}
  Also, for all $p \in \splxt$,
  \begin{equation*}
    \circrad{\splxtp} < 2\samconst.    
  \end{equation*}
\end{lem}

For convenience, we restate the consequences of
\Lemref{lem:non.protection.implies.forbidden} in terms of the
algorithmically convenient property guaranteed by
\Lemref{lem:clean.bad.hoop}, together with a couple of other
properties that are a direct consequence of 
\Defref{def:forbidden.config}. In particular, if $\splxt$ is a
\dgconfig, then it follows directly from \Defref{def:forbidden.config}
that
\begin{equation*}
  \longedge{\splxt} < (2 + \delta_0\psparseconst)\psamconst.
\end{equation*}
From this observation, and \Lemref{lem:perturb.Delone},
we obtain the diameter bound \ref{hyp:diam.bnd} below. 
\begin{thm}[Properties of forbidden configurations]
  \label{thm:no.hoops.implies.protection}
  \label{thm:prop.forbid.cfg} 
  Suppose that $\pts \subset \rem$ is a \ueset\ and that $\ppts$ is a
  perturbation of $\pts$ such that there is no simplex $\splxt \subset
  \ppts$ that satisfies \emph{all} of the following properties:
  \begin{enumerate}[label={$\mathcal{P}$\arabic*}]
  \item \label{hyp:clean.hoop.bnd}
    Simplex $\splxt$ has the \hoop\ property, with
    $\hoopbnd = 2 \left( \frac{16}{\sparseconst} \right)^3 \flakebnd$.
  \item \label{hyp:clean.facet.rad.bnd}
    For all $p \in \splxt$,
    $\circrad{\splxtp} < 2\samconst$.
  \item \label{hyp:diam.bnd} 
   $\longedge{\splxt} < \frac{5}{2}(1 + \frac{1}{2}\delta_0\sparseconst)\samconst$.
  \item \label{hyp:good.facets}
    Every facet of $\splxt$ is $\flakebnd$-good.
  \end{enumerate}

  If
  \begin{equation}
    \label{eq:dnobnd}
    \delta_0 \leq \flakebnd^{m+1} \quad \text{and} \quad \flakebnd
    \leq \frac{2\sparseconst^2 }{75},
  \end{equation}
  then $\ppts$ contains no \dgconfig s, and thus all the $m$-simplices
  in $\rdelppts$ are $\flakebnd$-good and $\delta$-protected, with
  $\delta = \delta_0 \psparseconst \psamconst$.
\end{thm}
In order to eliminate forbidden configurations, we only need to ensure that any one of the four properties of \Thmref{thm:prop.forbid.cfg} cannot occur in any simplex. As discussed in \Remref{rem:bad.good.facet} below, the algorithm does not exploit \ref{hyp:good.facets}, and only partially exploits \ref{hyp:clean.facet.rad.bnd}.

%% file: lem_hoop.tex
%

\newcommand{\tp}{\tilde{p}}
\newcommand{\Staup}{\circsphere{\opface{p}{\splxt}}}
\newcommand{\tdno}{\tilde{\delta}_0}

\subsection{Proof of the Hoop Lemma}
\label{sec:hoop.lem.proof}

In this appendix we demonstrate the Hoop Lemma~\ref{lem:hoop}, which
can be stated in full detail as:
\begin{lem}[Hoop Lemma]
  \label{lem:shell.lem}
  Let $\splxt$ be a $(k+1)$-dimensional \dgconfig\ in a \pueset. If
  \begin{equation*}
    \delta_0 \leq \frac{\psparseconst^2 \flakebnd^k}{6},
  \end{equation*}
  then for any $p \in \splxt$
  \begin{equation*}
    \distEm{p}{\circsphere{\splxtp}} \leq
    \left( \frac{84}{\psparseconst^3} \frac{\delta_0}{\flakebnd^k}
      + \frac{216}{\psparseconst^3} \flakebnd \right)
    \circrad{\splxtp}, 
  \end{equation*}
  and
  \begin{equation*}
    \circrad{\splxtp} < \left(1 +  \frac{3\delta_0}{\psparseconst
        \flakebnd^k} \right) \psamconst. 
  \end{equation*}
\end{lem}

Recall that \Lemref{lem:close.to.some.sphere} demonstrated that any
vertex in a \dgconfig\ lies close to a circumscribing sphere for its
opposing face.  We now use the fact that a \dgconfig\ is a flake to
bound the distance from a vertex to the circumsphere of its opposing
face. We employ the following characterisation of the altitudes of a
triangle:
\begin{lem}[Triangle altitude bound]
  \label{lem:tri.alt.bnd}
  For any non-degenerate triangle $\zeta = \simplex{\tp,u,v}$, we have
  \begin{equation*}
    \splxalt{\tp}{\zeta} = \frac{\norm{\tp-v}\norm{\tp-u}}{2\circrad{\zeta}}.
  \end{equation*}
\end{lem}
\begin{proof}
 Let $\alpha = \angle \tp uv$ and observe that
  \begin{equation*}
    \sin \alpha = \frac{\norm{\tp-v}}{2\circrad{\zeta}}.
  \end{equation*}
  Since $\splxalt{\tp}{\zeta} = \norm{\tp-u} \sin \alpha$, the result
  follows. 
\end{proof}

\begin{lem}[Distance to circumsphere]
  \label{lem:dist.to.circumsphere}
  Suppose $\splxt$ is a $\flakebnd$-flake with
  $\longedge{\splxt} \leq 3\psamconst$ and $\shortedge{\splxt} \geq
  \psparseconst \psamconst$.  If there exists a $p \in \splxt$ and a
  ball $B = \ballEm{C}{R}$ circumscribing $\splxtp$, with $R <
  \frac{3}{2}\psamconst$, and such that $\distEm{p}{\bdry{B}} \leq
  \tdno \shortedge{\splxtp}$ for some $\tdno \geq 0$, then
  $\distEm{p}{\Staup} \leq \hoopbnd \circrad{\splxtp}$, with
  \begin{equation*}
    \hoopbnd = \frac{14}{\psparseconst} \tdno +
    \frac{216}{\psparseconst^3} \flakebnd. 
  \end{equation*}
\end{lem}

\begin{figure}
  \begin{center}
    \includegraphics[width=.6\columnwidth]{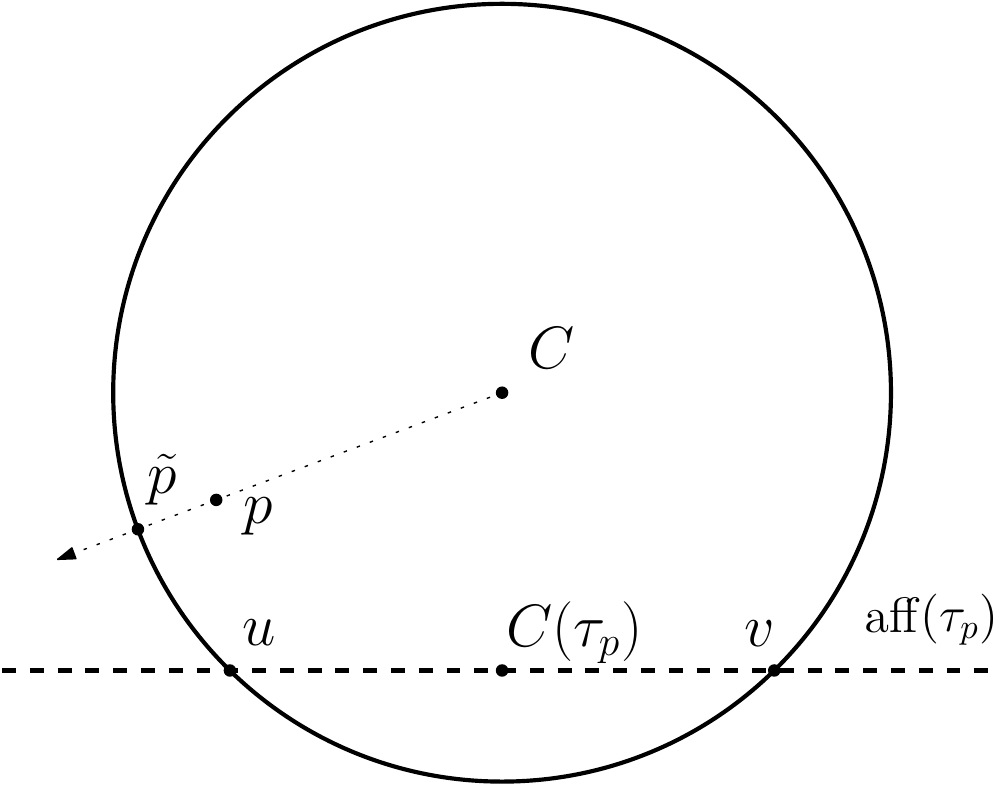} 
 \end{center}
 \caption{Diagram for \Lemref{lem:dist.to.circumsphere}.  }
 \label{fig:bound.to.Staup}
\end{figure}
\begin{proof}
  We are given that $p$ lies close to a circumscribing sphere $\bdry{B}$ for $\splxtp$. The fact that $\splxt$ is a flake implies that $p$ must also lie close to the affine hull of $\splxtp$. The result follows since $\Staup = \bdry{B} \cap \affhull{\splxtp}$. We quantify this by reducing the problem to two dimensions.

  Consider the plane $Q$ defined by $p$, $C$, and
  $\circcentre{\splxtp}$; if two of these three points coincide, we
  may choose $Q$ to be any plane which contains the three points. If
  $p=C$, then we have $\distEm{p}{\Staup} = R = \distEm{p}{\bdry{B}}
  \leq \tdno \shortedge{\splxtp} \leq \tdno 2 \circrad{\splxtp}$ which
  immediately implies the result. Thus suppose $p \neq C$. Let $\tp$ be
  the point of intersection of the ray from $C$ through $p$ with
  $\bdry{B}$, let $u \in \Staup \cap Q$ be the point closest to $\tp$,
  and let $v \in \Staup \cap Q$ be the farther point, as shown in
  \Figref{fig:bound.to.Staup}. Then
  \begin{equation}
    \label{eq:bound.to.Staup}
    \distEm{p}{u} \leq \distEm{p}{\tp} + \distEm{\tp}{u}.
  \end{equation}

  If $\tp = u \in \Staup$, then the result follows immediately, so we
  suppose these points to be distinct, and we consider the triangle
  $\zeta = \simplex{\tp,u,v}$. Since $\circrad{\zeta} = R$,
  \Lemref{lem:tri.alt.bnd} yields
  \begin{equation*}
    \distEm{\tp}{u} = \frac{2 R
      \splxalt{\tp}{\zeta}}{\distEm{\tp}{v}}.
  \end{equation*}
  Using our definition of $u$ we find
  \begin{equation*}
    \distEm{\tp}{v} \geq \frac{1}{2}\distEm{u}{v} = 
    \circrad{\splxtp} \geq \frac{1}{2} \shortedge{\splxtp}.   
  \end{equation*}
  The altitude is bounded by
  \begin{equation*}
    \begin{split}
      \splxalt{\tp}{\zeta}
      &\leq \distEm{\tp}{p} +
      \distEm{p}{\affhull{\seg{u}{v}}}\\
      &= \distEm{\tp}{p} + \splxalt{p}{\splxt}.
    \end{split}
  \end{equation*}
  Indeed, if $p^*$ is the orthogonal projection of $p$ into
  $\affhull{\splxtp}$, then $\seg{p}{p^*}$ is parallel to
  $\seg{C}{\circcentre{\splxtp}}$, because $\affhull{\splxtp}$ has
  codimension one in $\affhull{\splxt}$. It follows that $p^* \in
  Q \cap \affhull{\splxtp} = \affhull{\seg{u}{v}}$.

  By \Lemref{lem:flake.alt.bnd} 
  and the fact that $\longedge{\splxt} < 3\psamconst$, we have
  \begin{equation*}
    \splxalt{p}{\splxt} \leq  \frac{2\longedge{\splxt}^2
      \flakebnd} { \shortedge{\splxt}}
    \leq  \frac{6 \longedge{\splxt}
      \flakebnd} { \psparseconst }
    \leq  \frac{18 \flakebnd \shortedge{\splxtp}} { \psparseconst^2 }
  \end{equation*}
  Finally, recalling that $\distEm{p}{\tp} \leq \tdno
  \shortedge{\splxtp}$, and $R < \frac{3}{2}\psamconst$,
  we return to \Eqnref{eq:bound.to.Staup} and expand it using all of the
  subsequent displayed observations:
  \begin{equation}
    \label{eq:final.bound.to.Sw}
    \begin{split}
      \distEm{p}{u}
      &\leq \tdno \shortedge{\splxtp}
      + \frac{2 R \splxalt{\tp}{\zeta}}{\distEm{\tp}{v}}\\
      &\leq \tdno \shortedge{\splxtp} + \frac{4R}{\shortedge{\splxtp}}
      \left(\tdno \shortedge{\splxtp}
        + \frac{18 \flakebnd}{\psparseconst^2}
        \shortedge{\splxtp} \right)\\ 
      &< \tdno 2\circrad{\splxtp} + \frac{12}{\psparseconst}
      \left(\tdno  + \frac{18 \flakebnd}{\psparseconst^2}  \right)
      \circrad{\splxtp}\\
      &\leq \left( \frac{14}{\psparseconst} \tdno
        + \frac{216}{\psparseconst^3} \flakebnd  \right)
      \circrad{\splxtp}. 
    \end{split}
  \end{equation}
\end{proof}

\noindent
\textbf{Proof of \Lemref{lem:shell.lem}.}
  Using \Lemref{lem:close.to.some.sphere}, we  apply
  \Lemref{lem:dist.to.circumsphere} with
  \begin{equation*}
    \tdno =  \frac{6\delta_0}{\psparseconst^2 \flakebnd^k}.
  \end{equation*}
\noproof

%% file: algo.tex
%

\section{Algorithm}
\label{sec:algorithm}

In this section we present the algorithm. We start, in
\Secref{sec:main.result}, by announcing the guarantees of the
algorithm as our main theorem. 

\subsection{Main result}
\label{sec:main.result}

The goal and primary contribution of this paper is the presentation of
the perturbation \Algref{alg1}, and the demonstration of its
guarantees. 

In our analysis we employ three positive parameters, $\delta_0$,
$\flakebnd$, and $\pertbnd$, which are logically distinct. The
parameter $\delta_0$ specifies the protection that will be guaranteed
for the Delaunay $m$-simplices in $\rdelppts$, and $\flakebnd$ is a
bound on the quality of these simplices.  The analysis places an upper
bound on $\delta_0$ with respect to $\flakebnd$, and so for the
statement of our results, and the description of \Algref{alg1}, it is
convenient to combine the parameters by setting $\delta_0$ to be equal
to this upper bound:
\begin{equation*}
  \delta_0 = \flakebnd^{m+1}.
\end{equation*}
Our primary interest is in $\delta_0$, but it is more convenient to
express the results in terms of $\flakebnd$.  The analysis also places
an upper bound on $\flakebnd$ with respect to the parameter $\pertbnd$
that governs the amount of perturbation the input points may be
subjected to. We fix $\flakebnd$ with respect to this upper
bound, and let $\pertbnd$ be the only free parameter for the
algorithm.

The following theorem is demonstrated in \Secref{sec:analysis} and is
stated in full generality as \Thmref{thm:raw.main}:
\begin{thm}[Main result]
  \label{thm-main-theorem-of-the-paper}
  Taking as input a \ueset\ $\pts \subset \rem$, where
  $\sparseconst$ and $\samconst$ are known, and a positive parameter
  $\pertbnd \leq \frac{\sparseconst}{4}$,
  \Algref{alg1} produces a \pueset\
  $\ppts$ that is a $\pertbnd\samconst$-perturbation of $\pts$ such that all
  the Delaunay $m$-simplices in $\rdelppts$ are $\flakebnd$-good and
  $\delta$-protected, with
 \begin{equation*}
   \flakebnd = \frac{\pertbnd}{C},
   \quad \text{and} \quad 
   \delta = \flakebnd^{m+1}\psparseconst \psamconst,
  \end{equation*}
  where $C = \left( \frac{2}{\sparseconst} \right)^{3m^2 + 5m + 17}$,
  and $\psparseconst = \frac{\sparseconst-2\pertbnd}{1+ \pertbnd}$, and
  $\psamconst = (1+\pertbnd)\samconst$.

  The expected time complexity is
  \begin{equation*}
    \bigo{m}(\card{\pts})^2 + \left(\frac{2}{\sparseconst}\right)^{\bigo{m^2}}\card{\pts},
  \end{equation*}
  where the constant in the big-$O$ notation is an absolute constant.
\end{thm}
Although we require knowledge of two sampling parameters, $\sparseconst$, and $\samconst$, in practice one is easily deduced from the other by finding the minimum distance between two  points in $\pts$, and using the relation $\distEm{p}{q} \geq \sparseconst \samconst$.

We recall that by itself $\delta_0 = \flakebnd^{m+1}$ guarantees a lower
thickness bound proportional to $\delta_0^2 = \flakebnd^{2m+2}$ on the
Delaunay $m$-simplices \cite[Theorem 3.11]{boissonnat2013stab1}, but
this is much smaller than the $\flakebnd^m$ thickness guaranteed by
\Thmref{thm-main-theorem-of-the-paper}. If we were to set $\delta_0 =
0$ we would have a ``sliver exudation'' algorithm which would not
guarantee any $\delta$-genericity, but 
$\flakebnd$ would only increase by a factor of two.

\subsection{Algorithm overview}

We present an algorithm that will perturb an input \ueset\ $\pts$ to
obtain a \pueset\ $\ppts$ which contains no \dgconfigs.  The algorithm
takes as input a finite \ueset\ $\pts = \{p_1, \ldots, p_{n} \}
\subset \rem$. The output is obtained after $n$ iterations, such that
at the $i^{\text{th}}$ iteration a perturbation $\pts_i = \{p'_1,
\ldots, p'_i, p_{i+1}, \ldots, p_n \}$ is produced by perturbing the
point $p_i \mapsto p'_i$ in a way that ensures that there are no
\dgconfig s incident to $p'_i$ in $\pts_i$. Thus we have a sequence of
perturbations
\begin{equation*}
  \pts = \pts_0 \to \pts_1 \to \cdots \to \pts_n,
\end{equation*}
such that for all $i \in [1, \ldots, n]$, $\pts_i$ is a perturbation
of $\pts$ as well as of $\pts_{i-1}$, and $\pts_{i-1} \setminus
\{p_i\} = \pts_i \setminus \{p'_i\}$. Thus all the sets $\pts_i$ are
\pueset s.

At the $i^{th}$ iteration of the algorithm, all the points $p_{1}$ to
$p_{i-1}$, have already been perturbed, and the points $p_{i}$ to
$p_{n}$ have not yet been perturbed.  Using a uniform distribution, we
pick a random point $x \in B(p_{i}, \pertbnd\epsilon)$.
\begin{de}
  \label{def:good.perturbation}
  We say that $x$ is a \defn{good perturbation} of $p_i$ if for all
  simplices $\splxs \in \pts_{i-1} \setminus \{ p_i \}$, the
  simplex $\splxjoin{x}{\splxs}$ is not a \dgconfig.
\end{de}
If $x$ is a good perturbation of $p_i$, we let $p'_i = x$ and go on to
the next iteration, otherwise we choose a new random point from
$\ballEm{p_i}{\pertbnd\samconst}$.  The algorithm for determining if
$x$ is a good perturbation is discussed in
\Secref{sec:good.perturbations}, and the existence of good
perturbations is established in \Secref{sec:analysis}. The essential
ingredient is the \hoop\ property, and especially the symmetric nature
of this property.

The algorithm is shown in pseudocode in Algorithm~\ref{alg1}. Since a
good perturbation $p \mapsto p'$ ensures that there are no \dgconfig s
incident to $p'$ in the current point set, and in particular that no
new \dgconfig s are created, the output of the algorithm cannot
contain any \dgconfigs:
\begin{lem}
  \label{lem:no.forbidden.output}
  After the $i^{\text{th}}$ iteration of the algorithm, there are no
  \dgconfig s in $\cpltcplx{\pts_i}$ incident to $p'_j \in \pts_i$
  for any $j \in [1, \ldots, i]$. In particular, when the $n^{\text{th}}$
  iteration is completed, $\pts_n$ contains no \dgconfigs.
\end{lem}
\begin{proof}
  By the definition of a good perturbation, there is no \dgconfig\ 
  incident to $p_1 \in \pts_1$ after the first iteration has
  completed. Assume that at the $i^{\text{th}}$ iteration there are no
  \dgconfig s in $\pts_{i-1}$ incident to any $p'_j \in \pts_{i-1}$ for
  all $j < i$. At the completion of the $i^{\text{th}}$ iteration
  $\pts_{i-1} \setminus \{p_i\} = \pts_i \setminus \{p'_i\}$, so if
  there is a \dgconfig\ $\splxt \subset \pts_i$ that includes a $p'_j$
  with $j < i$, then $\splxt$ must also include $p'_i$, since
  otherwise we would have $\splxt \subset \pts_{i-1}$. But this
  contradicts the fact that $p'_i$ was chosen to be a good
  perturbation of $p_i$, thus establishing the claim.
\end{proof}

\begin{algorithm}
  \caption{Randomized perturbation algorithm}
  \label{alg1}
  \begin{algorithmic}
    \STATE{\rm Input:}\quad \ueset\ $\pts_0 = \{p_{1}, \dots, p_{n}\}
    \subset \rem$ and $\pertbnd$ 
   \FOR{$i = 1$ \TO $n$}
        \STATE ${\rm Flag} \gets 0$
        \STATE $x \gets p_i$
        \WHILE{${\rm Flag} \neq 1$}
            \IF{\texttt{good\_perturbation}$(x,p_i,\pts_{i-1})$}
                \STATE $p'_{i} \gets x$
                \STATE $\pts_i \gets ( \pts_{i-1} \setminus \{p_i\} ) \cup \{p'_i\}$
                \STATE ${\rm Flag} \gets 1$
            \ELSE
            \STATE // \texttt{random\_point}$(B(p_{i}, \pertbnd\epsilon))$
            outputs a point from the uniform distribution on $B(p_{i},
            \pertbnd\epsilon)$  
            \STATE $x \gets$ \texttt{random\_point}$(B(p_{i},
            \pertbnd\epsilon))$ 
            \ENDIF
        \ENDWHILE
    \ENDFOR
    \STATE // $\pts_{n} = \{ p_{1}', \, \dots, \, p_{n}'\}$, a
    $\delta$-generic \pueset, as described in
    \Thmref{thm-main-theorem-of-the-paper} 
    \STATE{\rm Output:}\quad $\pts_{n}$
  \end{algorithmic}
\end{algorithm}

\subsection{Implementation of good perturbations}
\label{sec:good.perturbations}

The geometric computations of the algorithm occur in the
\texttt{good\_perturbation} procedure, which is outlined in
\Algref{alg:good.perturbation}.  
%
The check for a good perturbation is a local operation.
We first establish a bound on the number of possible distinct
\dgconfig s incident to $p'$ in a perturbation $\ppts$ of
$\pts$. The first step is to bound the radius of a ball centred on $p$
that contains all such \dgconfigs:
\begin{lem}
  \label{lem:bound.forbid.rad}
  Suppose $\ppts$ is a perturbation of $\pts$, and $\splxt \subset
  \ppts$ is a \dgconfig, with $\delta_0 \leq \frac{2}{5}$.  If $p \in
  \pts$ and $p \mapsto p' \in \splxt$, then all the vertices of
  $\splxt$ originate from elements of $\pts$ contained in the ball
  $\ballEm{p}{r}$, with $r = (3 + \frac{\sparseconst}{2})\samconst$.
\end{lem}
\begin{proof}
  Suppose $q' \in \splxt$ originates from $q \in \pts$. Then, using
  Property~\ref{hyp:diam.bnd} and the perturbation
  bound~\eqref{eq:pertbnd}, the triangle inequality yields
  \begin{equation*}
   \begin{split}
    \distEm{p}{q}
    &\leq \longedge{\splxt} + \distEm{p}{p'} + \distEm{q}{q'}\\
    &< \frac{5}{2}(1 +  \frac{1}{2}\delta_0 \sparseconst)\samconst + 2\pertbnd
    \samconst\\ 
     &\leq (3 + \frac{1}{2}\sparseconst)\samconst.
    \end{split}
 \end{equation*}
\end{proof}

We exploit \Lemref{lem:bound.forbid.rad} to define the local
structures in which we check for \dgconfigs.  For any
point $p \in \pts$, let
\begin{equation*}
  \mathcal{N}_{p} = \ballEm{p}{(3 + \frac{\sparseconst}{2})\samconst}
  \cap \pts \setminus \{p\}, 
\end{equation*}
and define $\mathcal{S}_p$ to be the $m$-skeleton
of the complete complex on $\mathcal{N}_p$. In other words,
$\mathcal{S}_p$ consists of all $j$-simplices with vertices in
$\mathcal{N}_p$ and $j \leq m$. 

We let $\mathcal{S}_{p_i}(\pts_{i-1})$ denote the simplices in
$\pts_{i-1}$ that correspond to simplices in $\mathcal{S}_{p_i}$.  If
$\splxs' \in \pts_{i-1} \setminus \{p_i\}$ is such that it forms a
\dgconfig\ with $x \in \ballEm{p_i}{\pertbnd \samconst}$, then
$\splxs'$ belongs to $\mathcal{S}_{p_i}(\pts_{i-1})$.

\begin{algorithm}[ht] 
  \caption{\texttt{good\_perturbation}$(x,p, \ppts)$}
  \label{alg:good.perturbation}
  \begin{algorithmic}[1]
    \STATE // Test if $x$ is a good perturbation of $p$ in $\ppts$.
    \STATE // $\mathcal{S}_p(\ppts)$ is defined in
    \Secref{sec:good.perturbations}, and $\hoopbnd$
    is defined by Property~\ref{hyp:clean.hoop.bnd} of
    \Thmref{thm:prop.forbid.cfg}. 
    \STATE compute $\mathcal{S}_p(\ppts)$
    \FOR {each $\splxs \in \mathcal{S}_p(\ppts)$}
    \IF {$\circrad{\splxs} < \infty$} \label{lin:radbnd}
    \IF {
      $\abs{\distEm{x}{\circcentre{\splxs}} - \circrad{\splxs}} \leq 
        \hoopbnd 2\samconst$} \label{lin:hoopbnd}
    \RETURN \FALSE
    \ENDIF
    \ENDIF
    \ENDFOR
    \RETURN \TRUE
  \end{algorithmic}
\end{algorithm}
\Algref{alg:good.perturbation} reveals that Algorithm~\ref{alg1} uses
two geometric predicates: (1) a distance comparison (to compute
$\mathcal{S}_p(\ppts)$), and (2) the in-sphere tests implicit in
Line~\ref{lin:hoopbnd} of \Algref{alg:good.perturbation}. The
complexity of the algorithm will be discussed in
\Secref{ssec-algorithm-complexity}.
\begin{remk}
  \label{rem:bad.good.facet}
   We observe that \texttt{good\_perturbation} does not explicitly
   exploit Property~\ref{hyp:good.facets} of \dgconfig s. Also,
   Property~\ref{hyp:clean.facet.rad.bnd} is only really used for the
   bound on the right hand side of the inequality of
   Line~\ref{lin:hoopbnd}.  The volumetric analysis presented in
   \Secref{sec:analysis} counts all simplices $\splxs$ that could be a
   facet of a simplex with diameter bounded by
   Property~\ref{hyp:diam.bnd}, without consideration of the
   circumradius or thickness of $\splxs$. However, Properties
   \ref{hyp:good.facets} and \ref{hyp:clean.facet.rad.bnd} may be
   important in applications, and Line~\ref{lin:radbnd} serves as a
   reminder that they may be taken into account.
\end{remk}

%% file: analysis.tex
%

\section{Analysis of the algorithm}
\label{sec:analysis}

In this section we will prove
\Thmref{thm-main-theorem-of-the-paper}. We begin with a calculation of
the number of simplices contained in the local complexes
$\mathcal{S}_p(\ppts)$. Then in \Secref{sec:good.perts.exist}, following a standard practice in the analysis of perturbation algorithms~\cite{edelsbrunner2000smoothing,halperin2004controlled}, we
perform the volume calculations that show the existence of good
perturbations, and the probability of finding one with a random
point. Then in \Secref{ssec-algorithm-complexity} we analyse the
complexity and precision required by the algorithm.

\begin{lem}
  \label{lem-analysis-size-of-the-np-complex}
  Let $\pts \subset \rem$ be a \ueset.  For all $p \in \pts$, we have
  $\card{\mathcal{N}_{p}} \leq E_{1} \stackrel{{\rm def}}{=}
  \left( \frac{8}{\sparseconst} \right)^m$, and
  \begin{equation*}
    \card{\mathcal{S}_{p}} < E
    \stackrel{{\rm def}}{=} 
    2\left( \frac{8}{\sparseconst} \right)^{m^2 +m}. 
  \end{equation*}
\end{lem}
\begin{proof}
  In order to bound $\card{\mathcal{N}_p}$ we will use a packing
  argument in the ball $\ballEm{p}{(3 +
    \frac{\sparseconst}{2})\samconst}$ described in
  \Lemref{lem:bound.forbid.rad}. We extend the radius by the packing
  radius $r = \frac{\sparseconst \samconst}{2}$ of $\pts$. Thus let $R
  = (3 + \sparseconst)\epsilon$. It follows then that for any $p \in
  \pts$
  \begin{equation*}
    \card{\mathcal{N}_{p}} \leq  \left(\frac{R}{r}\right)^{m} =
    \left( \frac{2}{\sparseconst} \left( 3 + \sparseconst \right) \right)^m
    \leq \left( \frac{8}{\sparseconst} \right)^m = E_1.
  \end{equation*}
  This implies that for all $p\in \pts$,
  \begin{equation*}
  \card{\mathcal {S}_{p}} \leq \sum_{j=1}^{m+1} E_{1}^{j} <
  2E_1^{m+1} \leq  2\left( \frac{8}{\sparseconst} \right)^{m^2 +m} =   E.
  \end{equation*}
\end{proof}

\subsection{Existence of good perturbations}
\label{sec:good.perts.exist}

Recall that for any simplex $\splxs$ with $\circrad{\splxs} < \infty$
the circumsphere $\circsphere{\splxs}$ is contained in the diametric
sphere $\diasphere{\splxs}$. Thus if $\dist{x}{\diasphere{\splxs}} >
\hoopbnd \circrad{\splxs}$, then $\dist{x}{\circsphere{\splxs}} >
\hoopbnd \circrad{\splxs}$, and $\splxt = \splxjoin{x}{\splxs}$ cannot
have the \hoop\ property. As discussed below, it is convenient to use
$\diasphere{\splxs}$ instead of $\circsphere{\splxs}$, and there is
little cost since these objects coincide when $\splxs$ is an
$m$-simplex, and this dominates the calculation we are about to
describe.

The \texttt{good\_perturbation} procedure uses this sufficient
criterion to filter for good perturbations. The probability of
successfully finding a good perturbation by choosing a random point is
based on a volume calculation. Specifically, exploiting Properties
\ref{hyp:clean.hoop.bnd} and \ref{hyp:clean.facet.rad.bnd} of
\dgconfig s described in \Thmref{thm:prop.forbid.cfg}, we define the
\defn{forbidden volume} $F_p(\splxs)$ for $p$ contributed by $\splxs$
as the volume occupied in the perturbation ball
$\ballEm{p}{\pertconst}$ for $p$ consisting of those points that are
within a distance $\hoopbnd 2 \samconst$ from $\diasphere{\splxs}$, as
depicted in \Figref{fig:forbidden.volume}.

\begin{figure}[ht]
  \begin{center}
    \includegraphics[width=0.6\linewidth]{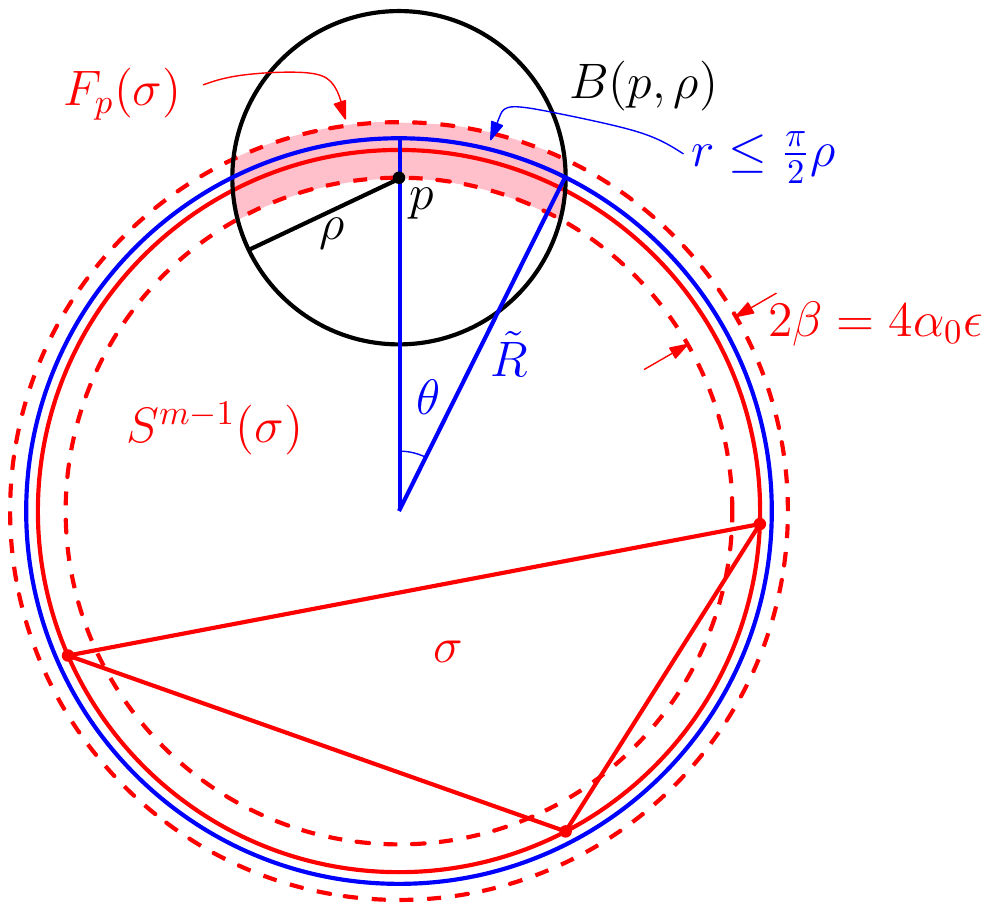} 
  \end{center}
  \caption[Forbidden volume]{The \defn{forbidden volume} $F_p(\splxs)$
    that a simplex $\splxs$ removes from the perturbation ball
    $\ballEm{p}{\pertconst}$ constitutes the points in
    $\ballEm{p}{\pertconst}$ that are within a distance $\hoopbnd
    2\samconst$ from $\diasphere{\splxs}$, as suggested by Properties
    \ref{hyp:clean.hoop.bnd} and \ref{hyp:clean.facet.rad.bnd} of
    \Thmref{thm:prop.forbid.cfg}.}
  \label{fig:forbidden.volume}
\end{figure}
 We let $\ballvol{j}$ denote the volume of a $j$-dimensional Euclidean
unit ball. The following lemma yields a bound on the forbidden volumes
$F_p(\splxs)$: 
\begin{lem}[Forbidden volume]
  \label{lem:shell.vol.bnd}
  If $S^{m-1}$ is a sphere of radius $R$ in $\rem$, then for any $p
  \in \rem$, and $\pertconst < R - \beta$, the volume
  $F_p(\pertconst, \beta, S^{m-1})$ of points contained in
  $\ballEm{p}{\pertconst}$, and within a distance $\beta$ from
  $S^{m-1}$ is bounded by
 \begin{equation*}
   F_p(\pertconst,\beta, S^{m-1}) \leq
   \ballvol{m-1}(\frac{\pi}{2}\pertconst)^{m-1} 2\beta. 
 \end{equation*}
\end{lem}
\begin{proof}
  Consider an $(m-1)$-sphere $S$, concentric with $S^{m-1}$ and with
  radius $\tilde{R}$ with $R - \beta \leq \tilde{R} \leq R +
  \beta$.  The intersection of $\ballEm{p}{\pertconst}$ with $S$
  will be a geodesic ball $\mathcal{B} \subset S$. Since $\pertconst <
  \tilde{R}$, the geodesic radius of $\mathcal{B}$, say
  $r=\tilde{R}\theta$, is subtended by an angle $\theta$ that is less
  than $\pi/2 $, and $\frac{2}{\pi}\theta \leq \sin \theta \leq
  \pertconst/\tilde{R}$. It follows that $r \leq
  \frac{\pi}{2}\pertconst$, independent of $R$ or $\tilde{R}$.

  Since the volume of a geodesic ball in an $(m-1)$-sphere is smaller
  than a Euclidean $(m-1)$-dimensional ball of the same
  radius~\cite[Theorem III.4.2]{chavel2006}, we have
  \begin{equation*}
    \vol(\mathcal{B}) \leq \ballvol{m-1}
    (\frac{\pi}{2}\pertconst)^{m-1}, 
  \end{equation*}
  and the stated bound follows.
\end{proof}
\begin{remk}
  If $\splxs \in \mathcal{S}_p(\ppts)$ is a $j$-simplex, with $j \leq
  m$, then it is also the face of many $m$-simplices in
  $\mathcal{S}_p(\ppts)$. Thus if $\dist{x}{\diasphere{\splxs}} \leq
  \hoopbnd 2 \samconst$, then we will also have
  $\dist{x}{\circsphere{\splxt}} \leq \hoopbnd 2 \samconst$ for any
  $m$-simplex $\splxt$ such that $\splxs \leq \splxt$. Thus the
  \texttt{good\_perturbation} \Algref{alg:good.perturbation} only
  really needs to consider the $m$-simplices in
  $\mathcal{S}_p(\ppts)$. This would save a factor of two in the
  estimate of $\card{\mathcal{S}_p}$, but if we wish to exploit
  Property~\ref{hyp:good.facets} of \Thmref{thm:prop.forbid.cfg}, as
  must be done in the context of finite precision, then all the lower
  dimensional simplices must also be taken into consideration. Indeed,
  if $\splxs$ is $\flakebnd$-good and has a small circumradius, we
  cannot assume that it is the face of an $m$-simplex with these
  properties.
\end{remk}

We now prove that at the $i$-th iteration of the algorithm there
exists a $p_{i}' \in B(p_{i}, \pertbnd\epsilon)$ that is a good
perturbation of $p_{i}$. We also establish an upper bound on the
expected number of times we have to pick random points from $B(p_{i},
\pertbnd\epsilon)$ in order to get a good perturbation. In the
description of the algorithm we let $\pertbnd$ determine $\delta_0$
and $\flakebnd$, but here we keep all three as separate parameters,
subject to constraint inequalities.
\begin{lem}[Existence of good perturbations]
  \label{lem:good.perturbation.existence}
  If
  \begin{equation}
    \label{eq:dno.and.gammaz.bnd}
   \delta_0 \leq \flakebnd^{m+1}, \quad \text{ and } \quad
    \flakebnd < \frac{\pertbnd}{K},
  \end{equation}
  where $K = \frac{\ballvol{m-1}}{\ballvolm} \left(
    \frac{8}{\sparseconst} \right)^{m^2} \left(
    \frac{16}{\sparseconst} \right)^{m+4}$, then at the
  $i^{\text{th}}$ iteration of the algorithm there exists a good
  perturbation $p_{i}'$ of $p_{i}$ such that no \dgconfig\ is incident
  to $p_{i}'$ in $\pts_i$, and the expected number of times we have to
  pick random points from $B(p_{i}, \pertbnd\epsilon)$ to get a good
  perturbation of $p_{i}$ is less than
  \begin{equation*}
    T = \frac{1}{1-\gamma},    
  \end{equation*}
  where
  \begin{equation*}
    \gamma = \frac{K \flakebnd}{\pertbnd}.
  \end{equation*}
\end{lem}
\begin{proof}
  We exploit \Thmref{thm:prop.forbid.cfg}.  Say that $x$ is a bad
  perturbation of $p \in \ppts$ if there is a $\splxs \in
  \mathcal{S}_p(\ppts)$ such that $\distEm{x}{\diasphere{\splxs}} \leq
  \hoopbnd 2\samconst$, with $\hoopbnd$ defined by
  Property~\ref{hyp:clean.hoop.bnd}.  Let $F_p(\splxs) :=
  F_p(\pertbnd\samconst, \hoopbnd 2\samconst, \diasphere{\splxs})$
  denote the volume in $\ballEm{p}{\pertbnd\samconst}$ that represents
  bad perturbations with respect to $\splxs$. Then
  \Lemref{lem:shell.vol.bnd} implies
    \begin{equation*}
      F_p(\splxs) \leq
      \ballvol{m-1}(\frac{\pi}{2})^{m-1} \pertbnd^{m-1}\samconst^{m-1}
      \hoopbnd 4\samconst. 
    \end{equation*}

    Using $E$ defined in \Lemref{lem-analysis-size-of-the-np-complex},
    we obtain a bound on $F_p$, the total volume of the bad
    perturbations in $\ballEm{p}{\pertbnd\samconst}$:
    \begin{align*}
      F_p &\leq E F_p(\splxs) &\\
      &\leq 8\left( \frac{8}{\sparseconst} \right)^{m^2 +m}
      \left(\frac{\pi}{2} \right)^{m-1}
      \hoopbnd   \pertbnd^{m-1}\ballvol{m-1}\samconst^m &\\
      &\leq 16\left( \frac{8}{\sparseconst} \right)^{m^2 +m}
      \left(\frac{\pi}{2} \right)^{m-1}
      \left(\frac{16}{\sparseconst} \right)^3 \flakebnd
      \pertbnd^{m-1}\ballvol{m-1}\samconst^m &
      \text{by Property~\ref{hyp:clean.hoop.bnd}} \\
      &\leq \left( \frac{8}{\sparseconst} \right)^{m^2}
      \left( \frac{16}{\sparseconst} \right)^{m + 4}
      \flakebnd \pertbnd^{m-1}\ballvol{m-1}\samconst^m &\\
    \end{align*}

  Therefore, the volume of the set of good perturbations of $p$ in
  $\ballEm{p}{\pertbnd \samconst}$ is greater than
  \begin{equation*}
    \ballvolm \pertbnd^{m}\epsilon^{m} - K\ballvolm \pertbnd^{m-1} \flakebnd
    \samconst^{m}, 
  \end{equation*}
  and it follows that the probability of getting a good perturbation
  of $p$ by a picking random point from $B(p, \pertbnd\epsilon)$ is
  greater than $1-\gamma$, where $\gamma = \frac{K
    \flakebnd}{\pertbnd}$.  Therefore the expected number of trials
  required to get a good perturbation is not greater than
  \begin{equation*}
    \sum_{i=0}^{\infty} (i+1)\gamma^{i}(1-\gamma) = \frac{1}{1-\gamma}.
  \end{equation*}
\end{proof}

\subsection{Complexity of the algorithm}
\label{ssec-algorithm-complexity}

Lemmas \ref{lem-analysis-size-of-the-np-complex} and
\ref{lem:good.perturbation.existence} lead directly to bounds on the
asymptotic properties of the algorithm:
\begin{lem}
  \label{lem:time.space.complexity}
  The expected time complexity of \Algref{alg1} is
  \begin{equation*}
    \bigo{m}(\card{\pts})^2 + 
    (1-\gamma)^{-1} \left( \frac{2}{\sparseconst} \right)^{\bigo{m^2}}
    \card{\pts}.
  \end{equation*}
  The space complexity required to run the algorithm is
  \begin{equation*}
    \left(\frac{2}{\sparseconst} \right)^{\bigo{m}} \card{\pts}
    + \left( \frac{2}{\sparseconst} \right)^{\bigo{m^2}}.
  \end{equation*}
\end{lem}
\begin{proof}
  The sets $\mathcal{N}_p$ can be computed by a na\"ive algorithm in
  $\bigo{m}(\card{\pts})^2$ time, while being stored in $\left(
    \frac{2}{\sparseconst} \right)^{\bigo{m}} \card{\pts}$ space,
  which is also sufficient to store the input and output point sets.

  The algorithm visits each point once, and it computes and stores the
  set $\mathcal{S}_p(\ppts)$ which has size $\left(
    \frac{2}{\sparseconst} \right)^{\bigo{m^2}}$. The
  \texttt{good\_perturbation} procedure
  (\Algref{alg:good.perturbation}) evaluates
  $\abs{\distEm{x}{\circcentre{\splxs}} - \circrad{\splxs}} \leq
  2\hoopbnd\samconst$ for every simplex $\splxs \in
  \mathcal{S}_p(\ppts)$.  This computation can be performed via
  determinant evaluations in $\bigo{m^3}$ time, so the time required
  to run the \texttt{good\_perturbation} algorithm is $\left(
    \frac{2}{\sparseconst} \right)^{\bigo{m^2}}$. The expected number
  of times it must be run on each point is $(1-\gamma)^{-1}$, and this
  yields the stated bound.
\end{proof}

\subsection{Summary of guarantees}

\Lemref{lem:no.forbidden.output} and
\Lemref{lem:good.perturbation.existence} guarantee that \Algref{alg1}
terminates with $\pts_n$ which contains no \dgconfig s and is a
perturbation of $\pts$. \Lemref{lem:time.space.complexity}
establishes the complexity bound.  Since 
Condition~\eqref{eq:dno.and.gammaz.bnd} demanded by
\Lemref{lem:good.perturbation.existence} implies 
Condition~\eqref{eq:dnobnd} required for
\Thmref{thm:no.hoops.implies.protection}, the main result
is established:
\begin{thm}[Main result]
  \label{thm:raw.main}
  \Algref{alg1} takes as input a \ueset\ $\pts \subset \rem$ and
  positive parameters $\pertbnd \leq \frac{\sparseconst}{4}$ and
  $\flakebnd$, with
  \begin{equation}
    \label{eq:raw.gamma.z.bnd}
    \flakebnd < \frac{\pertbnd}{K},
  \end{equation}
  where
  \begin{equation}
    \label{eq:raw.K}
    K = \frac{\ballvol{m-1}}{\ballvolm} \left(\frac{8}{\sparseconst}
    \right)^{m^2}\left(\frac{16}{\sparseconst} \right)^{m+4},     
  \end{equation}
  and $\ballvol{j}$ is the volume of the $j$-dimensional unit ball.

  By sequentially perturbing the points, it produces a \pueset\
  $\ppts$ that is a $\delta$-generic, $\pertbnd
  \samconst$-perturbation of $\pts$ and such that all the Delaunay
  $m$-simplices in $\rdelppts$ are $\flakebnd$-good and
 \begin{equation*}
    \delta = \flakebnd^{m+1}\psparseconst \psamconst,
  \end{equation*}
  where $\psparseconst$ and $\psamconst$ are defined in
  \Lemref{lem:perturb.Delone}. 

  The expected time complexity is less than
  \begin{equation*}
    \bigo{m}(\card{\pts})^2 + (1-\gamma)^{-1} \left(
      \frac{2}{\sparseconst} \right)^{O(m^{2})} \card{\pts}, 
  \end{equation*}
  where the constant in the big-$O$ notation is an absolute constant
  and
  \begin{equation*}
    \gamma = \frac{K \flakebnd}{\pertbnd}.
  \end{equation*}
\end{thm}
\Thmref{thm-main-theorem-of-the-paper} is a restatement of this
result, simplified by setting $\flakebnd = \frac{\pertbnd}{2K}$, and
by also observing that
\begin{equation}
  \label{eq:bound.sphere.vol.ratio}
  \frac{\ballvol{m-1}}{\ballvolm} \leq 2^m.
\end{equation}
Indeed, $\frac{\ballvol{m-1}}{\ballvolm}$ is a slowly growing function
of $m$, and the crude bound~\eqref{eq:bound.sphere.vol.ratio} can be
obtained from an elementary calculation using the
expression~\cite[Eq. (18), p. 9]{conway1988sphere} for $\log_2 V_m$.

The constant $K$ involved in the bound on $\flakebnd$ has been
computed explicitly, and cannot easily be reduced significantly. This
means that \Eqnref{eq:dno.and.gammaz.bnd}
yields a $2^{-\bigo{m^3}}$ bound on $\delta_0$, which results in very
small numbers, even in low dimensions.
Two of the powers of $m$ in the exponent come
from the consideration of all $m$-simplices in the neighbourhood of a
point (\Lemref{lem-analysis-size-of-the-np-complex}), and the other
comes from the dimension-gradated thickness bound introduced in the
\Defref{def:flake} of a flake. Analyses of traditional sliver
exudation algorithms suffer from similar tiny bounds, but in practice
these bounds appear to be pessimistic.

%% file: conclusions.tex
%

\section{Conclusions}
\label{sec:conclusions}

We have demonstrated an algorithm that will produce a $\delta$-generic
\pueset\ $\ppts$ that is a perturbation of a given \ueset\ 
$\pts$. The Delaunay triangulation of $\ppts$ is then quantifiably
stable with respect to changes in the metric or the points themselves.

Although our exposition assumes a finite set $\pts$, it is worth
observing that the analysis requires only local finiteness (the
intersection of $\pts$ with any compact set is a finite set), and the
algorithm extends trivially to the case of a periodic set $\tpts
\subset \rem$.  For example, we may have $\tpts = \tpts + v$ for any
$v \in \ints^m$, and $\tpts$ is $\samconst$-dense with respect to all
of $\rem$. In this framework we require that $\samconst < 1/2$, and we
may view $\tpts$ as a finite set $\pts$ in the standard flat torus
$\torus = \rem / \ints^m$. This has the advantage of avoiding boundary
considerations. It is also closer in spirit to the primary motivating
application of this work, which is the construction of Delaunay
triangulations of compact manifolds.

Funke et al.~\cite{funke2005} hinted at a much simpler analysis for
arguing that a perturbation of points in $\rem$, for arbitrary $m$,
has a good probability of being $\delta$-generic, with
$\flakebnd$-good simplicies. For a given point $p$, one simply
calculates the volumes of $\delta$-thick shells around the diametric
spheres of the nearby $m$-simplices (i.e., take $\beta=\delta$ in
\Figref{fig:forbidden.volume}), and one also accounts for the volumes
of ``slabs'' (i.e., the affine hull of each nearby $j$-simplex
thickened by an offset proportional to $\flakebnd^j$). The probability
that the perturbed point $p'$ violates the protection of a Delaunay
ball, or becomes the vertex of a $\flakebnd$-bad simplex, can thus be
made as small as required by appropriately reducing the size of
$\delta$ and $\flakebnd$, or by increasing the perturbation parameter
$\pertbnd$.

The problem with this simplified analysis is that although the
probability calculated for a given point depends only on points in a
neighbourhood (assuming a sampling density), these probabilities are
not independent. Conceptually, all the points must be perturbed at
once, and the probability of success is proportional to the total
number of points. Funke et al.~\cite[Section 4.3]{funke2005} mentioned
this limitation of their analysis. 

In this paper we have shown that the hoop property provides a way to
circumvent this difficulty and obtain a $\delta$-generic $\ppts$,
where $\delta/\samconst$ is only ultimately constrained by the separation
parameter $\sparseconst$, via Equations \ref{eq:pertbnd} and
\ref{eq:dno.and.gammaz.bnd}, and not by the sampling density or total
number of sample points. This is essential for our intended
application to meshing non-flat manifolds, which we have developed in other work~\cite{boissonnat2013manmesh.inria}. Building on the algorithm presented here, we give a constructive demonstration of the existence of Delaunay triangulations on compact abstract Riemannian manifolds. 

Thus we are already exploiting the theoretical benefits of the algorithm. The obstruction to a practical implementation is the computation required to verify that a perturbation is good. We are currently exploring an approach that avoids this problem by using only combinatorial tests and a result of Moser and Tardos~\cite{moser2010}.